\documentclass[draftcls,a4paper,11pt,onecolumn,journal]{IEEEtran}

\usepackage{cite}
\usepackage{hyphenat}
\usepackage{psfrag}
\usepackage{subfigure}
\usepackage[pdf]{graphicx}
\usepackage[T1]{fontenc}
\usepackage[utf8x]{inputenc}
\usepackage{amsmath,amsfonts,amsbsy,amssymb}
\usepackage{mathabx}
\usepackage{mathrsfs}
\usepackage[ruled]{algorithm}
\usepackage{algorithmic}
\usepackage[nolist]{acronym}
\usepackage{float}
\usepackage[colorlinks=false]{hyperref}
\usepackage{tipa}
\usepackage{dsfont}
\usepackage{tabularx}

\usepackage[absolute,overlay]{textpos}
\usepackage{tikz}
\usetikzlibrary{topaths}
\usetikzlibrary{shapes,trees,decorations}
\usetikzlibrary{arrows}
\usetikzlibrary[shadows]
\usetikzlibrary{positioning}
\usetikzlibrary{matrix}
\definecolor{darkblue}{rgb}{.1,.1,.6}

\usepackage[psfixbb,tightpage,displaymath,floats]{preview}
\PreviewEnvironment{tikzpicture}
\usepackage[nolists,tablesfirst,nomarkers]{endfloat}

\newtheorem{definition}{Definition}
\newtheorem{proposition}{Proposition}

\newtheorem{theorem}{Theorem}

\begin{document}

\title{Statistical Analysis of Self-Organizing Networks with \mbox{Biased} Cell
Association and Interference Avoidance}
\author{ {Carlos H. M. de Lima,~\IEEEmembership{Member, IEEE}, Mehdi
Bennis,~\IEEEmembership{Member, IEEE}, and Matti
Latva-aho,~\IEEEmembership{Senior Member, IEEE}\\ Email: \{carlosl, bennis,
matla\}@ee.oulu.fi}
\thanks{Authors are with the Centre for Wireless Communications (CWC),
University of Oulu, Finland.}
\thanks{Authors would like to thank the Finnish funding agency for technology
and innovation (Tekes), Elektrobit, Renesas Mobile, and Nokia Siemens Networks
for supporting this work.
This work has been also conducted in the framework of the ICT project
ICT-$4$-$248523$ BeFEMTO, which is partly funded by the EU.}}

\maketitle

\begin{acronym}[mmmmm]
	\acro{3GPP}[$3$GPP]{$3^\mathrm{rd}$ Generation Partnership Project}
	\acro{ABS}{Almost Blank Sub-frame}
	\acro{ADSL}{Asymmetric Digital Subscriber Line}
	\acro{DSL}{Digital Subscriber Line}
	\acro{ALBA-R}{Adaptive Load-Balanced Algorithm Rainbow}
	\acro{ALBA}{Adaptive Load-Balanced Algorithm}
	\acro{ALOHA}[ALOHA]{}
	\acro{APDL}{Average Packet Delivery Latency}
	\acro{AP}{Access Point}
	\acro{ASE}{Average Spectral Efficiency}
	\acro{BAP}{Blocked Access Protocol}
	\acro{BB}{Busy Burst}
	\acro{BC}{Broadcast Channel}
	\acro{BPP}{Binomial Point Process}
	\acro{BS}{Base Station}
	\acro{CAA}{Channel Access Algorithm}
	\acro{CAPEX}{Capital Expenditure}
	\acro{CAP}{Channel Access Protocol}
	\acro{CCDF}{Complementary Cumulative Distribution Function}
	\acro{CCI}{Co-Channel Interference}
	\acro{CDF}{Cumulative Distribution Function}
	\acro{CDMA}{Code Division Multiple Access}
	\acro{CDR}{Convex Lenses Decision Region}
	\acro{CF}{Characteristic Function}
	\acro{CGF}{Contention-based Geographic Forwarding}
	\acro{CM}{Coordination Mechanism}
	\acro{COMP}{Coordinated Multi-Point}
	\acro{CPICH}{Common Pilot Channel}
	\acro{CRA}{Conflict Resolution Algorithm}
	\acro{CRD}{Contention Resolution Delay}
	\acro{CRI}{Contention Resolution Interval}
	\acro{CRP}{Contention Resolution Protocol}
	\acro{CRS}{Cell-specific Reference Signal}
	\acro{CR}{Contention Resolution}
	\acro{CSMA/CA}{Carrier Sense Multiple Access with Collision Avoidance}
	\acro{CSMA/CD}{Carrier Sense Multiple Access with Collision Detection}
	\acro{CSMA}{Carrier Sense Multiple Access}
	\acro{CTM}{Capetanakis-Tsybakov-Mikhailov}
	\acro{CTS}{Clear To Send}
	\acro{D2D}{Device-To-Device}
	\acro{DAS}{Distributed Antenna System}
	\acro{DCF}{Distributed Coordination Function}
	\acro{DER}{Dynamic Exclusion Region}
	\acro{DHCP}{Dynamic Host Configuration Protoco}
	\acro{DL}{Downlink}
	\acro{E2E}{End-to-End}
	\acro{EDM}{Euclidean Distance Matrix}
	\acro{ES}{Evaluation Scenario}
	\acro{FAP}{Femto Access Point}
	\acro{FBS}{Femto Base Station}
	\acro{FDD}{Frequency Division Duplexing}
	\acro{FDM}{Frequency Division Multiplexing}
	\acro{FDR}{Forwarding Decision Region}
	\acro{FFR}{Fractional Frequency Reuse}
	\acro{FG}{Frequency Group}
	\acro{FPP}{First Passage Percolation}
	\acro{FUE}{Femto User Equipment}
	\acro{FU}{Femtocell User}
	\acro{GF}{Geographic Forwarding}
	\acro{GLIDER}{Gradient Landmark-Based Distributed Routing}
	\acro{GPSR}{Greedy Perimeter Stateless Routing}
	\acro{GeRaF}{Geographic Random Forwarding}
	\acro{HDR}[HDR]{High Data Rate}
	\acro{HII}{High Interference Indicator}
	\acro{HNB}{Home Node B}
	\acro{HN}[HetNet]{Heteronegeous Network}
	\acro{HOS}{Higher Order Statistics}
	\acro{HUE}{Home User Equipment}
	\acro{IMT}{International Mobile Telecommunications}
	\acro{ITU}{International Telecommunication Union}
	\acro{ICIC}{Inter-Cell Interference Coordination}
	\acro{IEEE}[IEEE]{}
	\acro{IP}{Interference Profile}
	\acro{KPI}{Key Performance Indicators}
	\acro{LN}{Log-Normal}
	\acro{LTE}{Long Term Evolution}
	\acro{LoS}{Line-of-Sight}
	\acro{M2M}{Machine-To-Machine}
	\acro{MACA}{Multiple Access with Collision Avoidance}
	\acro{MAC}{Medium Access Control}
	\acro{MBS}{Macro Base Station}
	\acro{MGF}{Moment Generating Function}
	\acro{MIMO}{Multiple-Input Multiple-Output}
	\acro{MPP}{Marked Point Process}
	\acro{MRC}{Maximum Ratio Combining}
	\acro{MS}{Mobile Station}
	\acro{MUE}{Macro User Equipment}
	\acro{MU}{Macrocell User}
	\acro{NB}{Node B}
	\acro{NLoS}{Non Line-of-Sight}
	\acro{NRT}{Non Real Time}
	\acro{OFDMA}{Orthogonal Frequency Division Multiple Access}
	\acro{OOP}{Object Oriented Programming}
	\acro{OPEX}{Operating Expenditure}
	\acro{OP}{Outage Probability}
	\acro{OS}{Order Statistic}
	\acro{PBS}{Pico Base Station}
	\acro{PC}{Power Control}
	\acro{PCI}{Physical Cell Indicator}
	\acro{PDF}{Probability Density Function}
	\acro{PDSR}{Packet Delivery Success Ratio}
	\acro{PGF}{Probability Generating Function}
	\acro{PMF}{Probability Mass Function}
	\acro{PPP}{Poisson Point Process}
	\acro{PP}{Point Process}
	\acro{PRM}{Poisson Random Measure}
	\acro{QoS}{Quality of Service}
	\acro{RAS}{Random Access System}
	\acro{RAT}{Radio Access Technology}
	\acro{RA}{Random Access}
	\acro{RCA}{Random Channel Access}
	\acro{RD}[R$\&$D]{Research $\&$ Development}
	\acro{REB}{Range Expansion Bias}
	\acro{RE}{Range Expansion}
	\acro{RF}{Radio Frequency}
	\acro{RIBF}{Regularized Incomplete Beta Function}
	\acro{RSSI}{Received Signal Strength Indicator}
	\acro{PSS}{Primary Synchronization Channel}
	\acro{SSS}{Secondary Synchronization Channel}
	\acro{RMA}{Random Multiple-Access}
	\acro{RN}{Relay Node}
	\acro{RNTP}{Relative Narrowband Transmit Power}
	\acro{RRM}{Radio Resource Management}
	\acro{RSA}{Relay Selection Algorithm}
	\acro{RSRP}{Reference Signal Received Power}
	\acro{RSS}{Received Signal Strength}
	\acro{RS}{Relay Selection}
	\acro{RTS}{Request to Send}
	\acro{RT}{Real Time}
	\acro{RV}{Random Variable}
	\acro{SC}{Selection Combining}
	\acro{SDR}{Sectoral Decision Region}
	\acro{SF}{Sub-Frame}
	\acro{SG}{Stochastic Geometry}
	\acro{SIC}{Successive Interference Cancellation}
	\acro{SINR}{Signal-to-Interference plus Noise Ratio}
	\acro{SIR}{Signal-to-Interference Ratio}
	\acro{SLN}{Shifted Log-Normal}
	\acro{SM}{State Machine}
	\acro{SMP}{Semi-Markov Process}
	\acro{SNR}{Signal to Noise Ratio}
	\acro{SON}{Self-Organizing Network}
	\acro{SPP}{Spatial Poisson Process}
	\acro{STA}{Standard Tree Algorithm}
	\acro{TAS}{Transmit Antenna Selection}
	\acro{TCP}{Transmission Control Protocol}
	\acro{TC}{Transmission Capacity}
	\acro{TDD}{Time Division Duplexing}	
	\acro{TDMA}{Time Division Multiple Access}	
	\acro{TS}{Terminal Station}
	\acro{TTI}{Transmission Time Interval}
	\acro{UDM}{Unit Disk Model}
	\acro{UD}{Unit Disk}
	\acro{UE}{User Equipment}
	\acro{ULUTRANSIM}[UL UTRANSim]{R6 Uplink UTRAN Simulator}
	\acro{UL}{Uplink}
	\acro{UML}{Unified Modeling Language}
	\acro{UMTS}{Universal Mobile Telecommunications System}
	\acro{WCDMA}{Wideband Code Division Multiple Access}
	\acro{WSN}{Wireless Sensor Network}
	\acro{iid}[\textup{i.i.d.}]{independent and identically distributed}
\end{acronym}

\begin{abstract}
	In this work, we assess the viability of heterogeneous networks composed of
	legacy macrocells which are underlaid with self-organizing picocells.
	Aiming to improve coverage, cell-edge throughput and overall system
	capacity, self-organizing solutions, such as range expansion bias, almost
	blank subframe and distributed antenna systems are considered. 
	Herein, stochastic geometry is used to model network deployments, while
	higher-order statistics through the cumulants concept is utilized to
	characterize the probability distribution of the received power and
	aggregate interference at the user of interest.
	A comprehensive analytical framework is introduced to evaluate the
	performance of such self-organizing networks in terms of outage probability
	and average channel capacity with respect to the tagged receiver.
	To conduct our studies, we consider a shadowed fading channel model
	incorporating log-normal shadowing and Nakagami-$m$ fading.
	Results show that the analytical framework matches well with numerical
	results obtained from Monte Carlo simulations.
	We also observed that by simply using almost blank subframes the aggregate
	interference at the tagged receiver is reduced by about $12\mathrm{dB}$.
	Although more elaborated interference control techniques such as, downlink
	bitmap and distributed antennas systems become needed, when the density of
	picocells in the underlaid tier gets high.
\end{abstract}

\acresetall

\section{Introduction}
\label{SEC:INTRODUCTION}
%
%
Targeting at upcoming releases, the \ac{3GPP} standardization body has focused
on enhancing the end-user satisfaction and performance of \ac{LTE} systems by
adopting new deployments strategies and concepts such as \acp{HN} and
self-organization.
In fact, legacy cellular systems with predefined structure and centralized
coordination cannot keep up with the stringent requirements of next generation
wireless systems, which demand high spectral efficiency and ubiquitous coverage
with fairness at cell border.
For instance, \ac{LTE}-Advanced aims at peak data rates up to
$1\,\mathrm{Gbps}$ which contrasts with current \ac{LTE} systems which deliver
at most $100\,\mathrm{Mbps}$ or even \ac{ADSL} technology over cooper landlines
that can transmit at $24\,\mathrm{Mbps}$ only.
Operators have indeed very few options available to meet such requirements:
increase the density of macrocell sites, but that hinges on regulatory studies
and approval; upgrade \ac{RAT} which takes time and do not fill the capacity
gap completely; or expand the radio spectrum resource, but that is definitely a
very expensive and lingering alternative.

In this context, heterogeneous deployments which underlay legacy macrocells
with low-cost, -power and -complexity small cells emerge as a promising and
inexpensive alternative to meet these strict requirements.
Future networks indeed benefit from self-organization in several situations,
for example, to cope with the uncertainties of random networks wherein moving
nodes need to communicate over volatile wireless channels; and to dynamically
reconfigure and maintain infrastructureless deployments of small cells with
large amount of nodes in which traditional and centralized methods become
costly or even unfeasible.
In order to tap into the full benefits of large-scale \acp{SON}, a number of
challenges still need to be tackled, including their deployment, operation,
automation and maintenance \cite{ART:PREHOFER-ICM05, ART:ALIU-ICST12}.

\subsection{Related Work}
\label{SEC:RELATED_WORK}
The design and implementation of self-organizing functionalities in \acp{HN} is
a topic of significant interest as evidenced by the number of recent
publications \cite{ART:PREHOFER-ICM05, ART:AKHTMAN-PROC10, ART:GUVENC-ICL11,
PROC:MUKHERJEE-ASILOMAR11, ART:PEREZ-JSAC12, ART:MUKHERJEE-JSAC12,
PROC:JO-GLOBECOM11}.
For instance, the self-organization concept is used to devise cognitive radio
resource management schemes to mitigate cross-tier interference and guarantee
users \ac{QoS} in distinct heterogeneous deployments scenarios
\cite{ART:LIEN-TWC11}.
More recently, the \ac{REB} concept is discussed within \ac{3GPP} as a baseline
solution to boost the offloading potential of heterogeneous deployments.
In that regard, Authors in \cite{PROC:OKINO-ICC11} investigate the cell range
expansion and interference mitigation in heterogeneous networks. 
Following the same lines, G\"{u}ven\c{c} instigates the capacity and fairness
of heterogeneous networks with range expansion and interference coordination
\cite{ART:GUVENC-ICL11}.
In \cite{PROC:JO-GLOBECOM11}, Jo \textit{et al.} use the \ac{SG} framework to
assess how the biased cell association procedure performs in heterogeneous
networks by means of the outage probability.
\textcolor{black}{In multi-tier heterogeneous networks where the locations of
\acp{BS} are modeled as independent \ac{PPP}, the joint distribution of the
downlink \ac{SINR} at the tagged receiver is derived when the serving \ac{BS}
is selected as either the nearest or the strongest with respect to the user
of interest \cite{PROC:MUKHERJEE-ICC12}.}

\subsection{Contributions and Organization}
\label{SEC:CONTRIBUTIONS AND ORGANIZATION}
In this work, we assess the performance of heterogeneous networks consisting of
legacy macrocells with underlaid small cells.
The offloading potential and self-organizing feature of small cells are studied
so as to increase the overall spectral efficiency and meet the requirements of
next generation systems as well.
After providing definitions and models in Section \ref{SEC:SYSTEM_MODEL}, a
comprehensive analytical framework which resorts to \ac{SG} and \ac{HOS} is
then introduced to evaluate the performance of such heterogeneous deployments.
We discuss the self-organizing solutions and network operation in Section
\ref{SEC:NETWORK_OPERATION}.
In that regard, we consider heterogeneous scenarios which employ \ac{REB} to
improve spatial reuse and balance load between tiers.
To cope with the resulting \ac{CCI}, \ac{ABS} is considered as the baseline
\ac{ICIC} technique.
Thereafter, a bitmap indicator, referred to as \ac{DL}-\ac{HII}, is used to
identify the dominant interferers and improve the \ac{SIR} at the receiver of
interest.
We then investigate the concept of virtual \ac{DAS} which is yet another
self-organization solution to mitigate interference and improve the received
signal at the receiver of interest.
Afterwards, practical evaluation scenarios are defined in Section \ref{SEC:CCI}
wherein the resulting cross-tier interference and practical mechanism to
mitigate it are of primary interest.
%
%
Numerical results are provided in Section \ref{SEC:PERFORMANCE_ANALYSIS}.
Finally, we draw conclusions and make final remarks in Section
\ref{SEC:CONCLUSIONS}.

\section{System Model and Analytical Framework}
\label{SEC:SYSTEM_MODEL}
In preparation for the description of the evaluation scenarios and their
performance analysis, we first present our assumptions, make definitions and
introduce our system models.

\subsection{Definitions and Notation}
\label{SEC:DEFINITIONS_AND_NOTATIONS}

\begin{definition}{\textbf{(Tagged receiver)}}
	The \ac{MU} who is taken as the reference to compute the aggregate \ac{CCI}
	and performance metrics on the \ac{DL} of the evaluation scenarios.
	In stochastic geometry, Palm distributions and the related Campbell's
	theorem are used to characterize a random pattern with respect to a typical
	point of the process, so that network-wide performance can be characterized
	by the average behavior of this ``tagged'' node
	\cite{BOOK:BADDELEY-SPRINGER03, ART:ANDREWS-CM10}.
	\label{DEF:TAGGED_RECEIVER}
\end{definition}

\begin{definition}{\textbf{(Observation region)}}
	An annular region around the tagged receiver over which we account for the
	aggregate interference.
	The observation region is denoted by $\mathcal{O}$ and defined by the
	minimum and maximum radii $R_m$ and $R_M$, respectively.
	\label{DEF:OBSERVATION_REGION}
\end{definition}

\begin{definition}{\textbf{(Partial moment of a random variable)}}
	Let $Y$ be a \ac{RV}, then $\operatorname{E}_Y^n \negthinspace \left [ y_m,
	y_M \right ] = \int_{y_m}^{y_M} {y^n f_Y \negthinspace \left ( y
	\right)\mathrm{d}y}$ denotes its $n^\mathrm{th}$ partial moment with $y_m$
	and $y_M$ indicating the lower and upper integration limits, respectively.
	\label{DEF:PARTIAL_MOMENT}
\end{definition}

\subsection{Propagation Channel Model}
\label{SEC:PHY_MODEL}
Radio links are degraded by path loss and shadowed fading, which is assumed to
be independent over distinct network entities and positions.
A signal strength decay function, $l\negthinspace\left( r \right) =
r^{-\alpha}$, where $\alpha$ is the path loss exponent, describes the path loss
attenuation (unbounded path loss model \cite{ART:INALTEKIN-JSAC09}), while the
received squared-envelop due to multi-path fading and shadowing is represented
by a \ac{RV} $X \in \mathbb{R}^+$ with \ac{CDF} and \ac{PDF} denoted by
$F_X\negthinspace\left( x \right)$ and $f_X\negthinspace\left( x \right)$,
respectively.
An arbitrary interferer disrupts the communication of the tagged receiver with
a component given by
\begin{align}
	Y = p\,\thinspace l\negthinspace\left( r \right) x,
	\label{EQ:INTERFERENCE_COMPONENT}
\end{align}
\noindent{where $p$ yields this interferer transmitted power, $r$ is the
separation distance from its position to the tagged receiver, and $x$ yields
the corresponding shadowed fading.}

The composite distribution of the received squared-envelop due to \ac{LN}
shadowing and Nakagami-$m$ fading has a Gamma-\ac{LN} distribution with
\ac{PDF} \cite{BOOK:STUBER-SPRINGER00},
\begin{align}
	f_X(x) = \int\limits_{0}^{\infty}\left(\frac{m}
	{\omega}\right)^m\frac{x^{m-1}} {\Gamma(m)}
	\exp\left(-\frac{m}{\omega}x\right)\times\frac{\xi}{\sqrt{2 \pi} \sigma
	\omega}\exp{\left[ -\frac{\left ( \xi\ln{\omega} -\mu_{\Omega_p}
	\right)^2}{2 \sigma_{\Omega_p}^2} \right]}\mathrm{d}\omega,
	\label{EQ:GAMA_LOGNORMAL_PDF}
\end{align}
{\noindent where $m$ is the shape parameter of the Gamma distribution, $\xi =
\ln \left( 10 \right)/10$, $\Omega_p$ is the mean squared-envelop,
$\mu_{\Omega_p}$  and $\sigma_{\Omega_p}$ is the mean and standard deviation of
$\Omega_p$, respectively.}

Ho \textit{et al.} show in \cite{ART:HO-ACM95} that a composite Gamma--\ac{LN}
distribution can be approximated by a single \ac{LN} distribution with mean and
variance (in logarithmic scale) given by $\mu_{\mathrm{dB}} = \xi \left [
\psi\left ( m \right ) - \ln\left ( m \right)\right ] + \mu_{\Omega_p}$ and
$\sigma_{\mathrm{dB}}^2 = \xi^{2}\zeta\left ( 2, m \right) +
\sigma_{\Omega_p}^2$, where $\psi\left ( m \right )$ is the Euler psi function
and $\zeta\left ( 2, m \right )$ is the generalized Riemann zeta function
\cite{BOOK:ABRAMOWITZ-DOVER03}.
In what follows, we use this single \ac{LN} approximation to characterize the
radio channel attenuations in various evaluation scenarios.

\subsection{Network Deployment Model}
\label{SEC:DEPLOYMENT_MODEL}
The \ac{DL} of a heterogeneous networks consisting of an umbrella \ac{MBS} and
an underlaid tier of self-organizing small cells is modeled.
We assume that \acp{MBS} follow centralized coordination and spectrum
allocation such that inter macrocell interference is mitigated, for example,
using fractional frequency reuse \cite{ART:PEREZ-JSAC12}.
In these scenarios, picocells are uniformly scattered over the network area,
while both tiers operate in \ac{TDD} mode and share the whole spectrum.
\textcolor{black}{In addition, communicating nodes are assumed to be
synchronized so that uplink and downlink transmissions do not interfere with
each other.}
In every \ac{SF}, each serving \ac{BS} schedules a single user terminal and
interference coordination are implemented in the time-domain.
The set of associated user terminals are also uniformly distributed within the
transmission range of their serving cells.
Nodes communicate using antennas with omni directional radiation pattern and
fixed power.
The macrocell tier transmits at a maximum power of $46\,\mathrm{dBm}$ and
picocells use $30\,\mathrm{dBm}$.

Active picocells constitute a homogeneous \ac{PPP} $\Phi$ with density
$\lambda$ in $\mathbb{R}^2$.
The number of picocells in an arbitrary region $\mathcal{R}$ of area $A$ is a
Poisson \ac{RV} with parameter $\lambda A$ \cite{ BOOK:KINGMAN-OXFORD93}.
Additionally, we assume the fading effect as a random mark associated with each
point of $\displaystyle {\Phi}$.
By virtue of the \textup{Marking theorem} \cite{BOOK:BADDELEY-SPRINGER03,
BOOK:KINGMAN-OXFORD93}, the resulting process,
\begin{align}
	{\widetilde{\Phi}} = \left\{ \left( \varphi, x \right); \varphi \in \Phi
	\right\},
	\label{EQ:MPP}
\end{align}
corresponds to a \ac{MPP} on the product space $\mathbb{R}^2 \times
\mathbb{R}^{+}$, whose random points $\varphi$ denoting transmitters locations
and belong to the stationary point process ${{\Phi}}$.

\subsection{\acl{HOS} and the \ac{LN} approximation}
\label{SEC:ANALYTICAL_FRAMEWORK}
We introduce our analytical framework which uses stochastic geometry to model
network deployments \cite{BOOK:KINGMAN-OXFORD93, BOOK:STOYAN-WILEY95}, and
higher order statistics through the cumulants concept to recover both the
distributions of the received power $Y$ and the aggregate \ac{CCI} $Z$ at the
tagged receiver \cite{BOOK:ABRAMOWITZ-DOVER03, ART:GHASEMI-JSAC08}.
The Slivnyak's theorem and its associated Palm probability are then used to
derive the aggregate \ac{CCI} and compute average performance figures
conditional on the location of the tagged receiver.

To establish this framework, we begin by applying Campbell's theorem
\cite{BOOK:KINGMAN-OXFORD93, BOOK:STOYAN-WILEY95} to the \ac{MPP}
$\widetilde{\Phi}$ defined in \eqref{EQ:MPP} so as to determine the \ac{CF} of
the distribution of the aggregate \ac{CCI}.
\begin{definition}
	Let $Z=\sum_{\left( \varphi, x \right) \in \widetilde{\Phi}} Y$ be a
	\ac{RV} representing the aggregate \ac{CCI} generated by the interfering
	process $\widetilde{\Phi}$, and $j=\sqrt{-1}$ be the imaginary unity; then,
	the function $\Psi:\mathbb{R} \rightarrow \mathbb{C}$ defined as,
	\begin{align}
	\Psi_Z\left ( \omega \right )=\operatorname{E}\left [ e^{ j\omega Z} \right
	],
	\label{EQ:CHARACTERISTIC_FUNCTION}
	\end{align}
	is called the \ac{CF} of $Z$.
\end{definition}

The corresponding $n^\textnormal{th}$ cumulants are obtained by computing
higher order derivatives of \eqref{EQ:CHARACTERISTIC_FUNCTION} as presented in
our next proposition \cite{BOOK:ABRAMOWITZ-DOVER03}.
\begin{proposition}
	Let $Z$ be a \ac{RV} and $\Psi_Z\left ( \omega \right )$ its \ac{CF}.
	Let $n \in \mathbb{N}$.
	Provided that the $n^{\text{th}}$ moment exists and is finite.
	Then, $\Psi_Z\left ( \omega \right )$ is differentiable $n$ times and 
	\begin{align}
		\kappa_{n} = \frac{1}{j^n}\left [ \frac{\partial^n }{\partial
		\omega^n}\ln\Psi_{Z} \left ( \omega \right )\right
		]_{\omega=0}\hspace{-1.75em}.
		\label{EQ:CUMULANT_PARTIAL_DERIVATIVE}
	\end{align}
\end{proposition}

\begin{IEEEproof}
	See \cite[Section 9.4]{BOOK:RESNICK-BIRKHAUSER1999}.
\end{IEEEproof}

Motivated by the fact that the density of $Z$ has no exact closed form
expression \cite{PROC:WU-GLOBECOM05} and that its distribution is heavy-tailed
and positively skewed \cite{ART:GHASEMI-JSAC08}, we use the \ac{LN}
approximation whose parameters are estimated from the cumulants of the
aggregate \ac{CCI}.
We relate the parameters of this equivalent \ac{LN} distribution to the
cumulants of the actual distribution of the aggregate \ac{CCI} as follows,
\begin{align}
	\mu = \ln{\left(\frac{\kappa_1^2}{\sqrt{\kappa_1^2 + \kappa_2}}\right)},
	\thickspace \text{and} \thickspace \sigma^2 = \ln{\left( 1 +
	\frac{\kappa_2}{\kappa_1^2} \right)}.
	\label{EQ:LOGNORMAL_PARAMETERS}
\end{align}
{\noindent where $\mu$ is the mean and $\sigma$ is the standard deviation of
the distribution $\mathsf{Normal}(\mu, \sigma^2)$ in the logarithmic scale.}

\section{Biased Cell Association and Handover Probability}
\label{SEC:REB}
Following the standard handover procedure \cite{PROC:DIMOU-VTC2009}, the tagged
\ac{MU} is transferred to the underlaid picocell tier only if the pilot signal
of the target \ac{PBS} is strictly higher than the umbrella
\ac{MBS}\footnote{We consider that migrating \acp{MU} are not affected by the
ping-pong effect and that the predefine triggering time has already elapsed
\cite{TR:3GPP-TS23.009}.} as follows,
\begin{align}
	Y^\mathrm{P} > Y^\mathrm{M} + \Omega,
	\label{EQ:STD_HO}
\end{align}
\noindent where the \ac{RV} $Y^\mathrm{P}$ refers to the power received from
the target \ac{PBS}, $Y^\mathrm{M}$ yields the power received from the umbrella
\ac{MBS} and $\Omega$ is the handover hysteresis to avoid the ping-pong effect
\cite{TR:3GPP-TS23.009}.

However, in most circumstances, the umbrella \ac{MBS} overpowers the underlaid
tier which shrinks the coverage of the small cells and compromises the expected
gains of spatial and frequency reuse \cite{PROC:OKINO-ICC11, ART:GUVENC-ICL11}.
In such large-scale heterogeneous deployments, transceivers have various
communication capabilities and the restrictive nature of the typical handover
procedure worsen the load unbalance problem across tiers.
To alleviate this problem, \acs{3GPP} suggests adding a positive bias $\Delta
\mathrm{REB}$ to the picocells received power so that the rate of \acp{MU}
handovers to the underlaid tier increases \cite{PROC:OKINO-ICC11} as given next
\begin{align}
	Y^\mathrm{P} + \Delta \mathrm{REB} > Y^\mathrm{M} + \Omega.
	\label{EQ:REB_HO}
\end{align}

Indeed, the \ac{REB} prompt the macrocell offloading and improves the spectral
efficiency by relaxing the standard association criteria used by \acp{MU}.
Unfortunately, by doing so, \acp{MU} within the expanded region of picocells do
not actually connect to the strongest \acp{BS} and are exposed to high
interference levels from the macrocell tier.
Fig. \ref{FIG:REB} illustrates the operation of the \ac{REB} concept in
heterogeneous scenarios composed of an umbrella macrocell and underlaid
picocells.
The coverage area of the target picocell is artificially increased by the
positive bias $\Delta \mathrm{REB}$ as indicated by the handover criterion in
\eqref{EQ:REB_HO}.
As a result, \acp{MU} are offloaded to the picocell tier more often and
unburden the umbrella macrocell.
\begin{figure}[h!]
	\centering
	\includegraphics[width=1.\columnwidth]{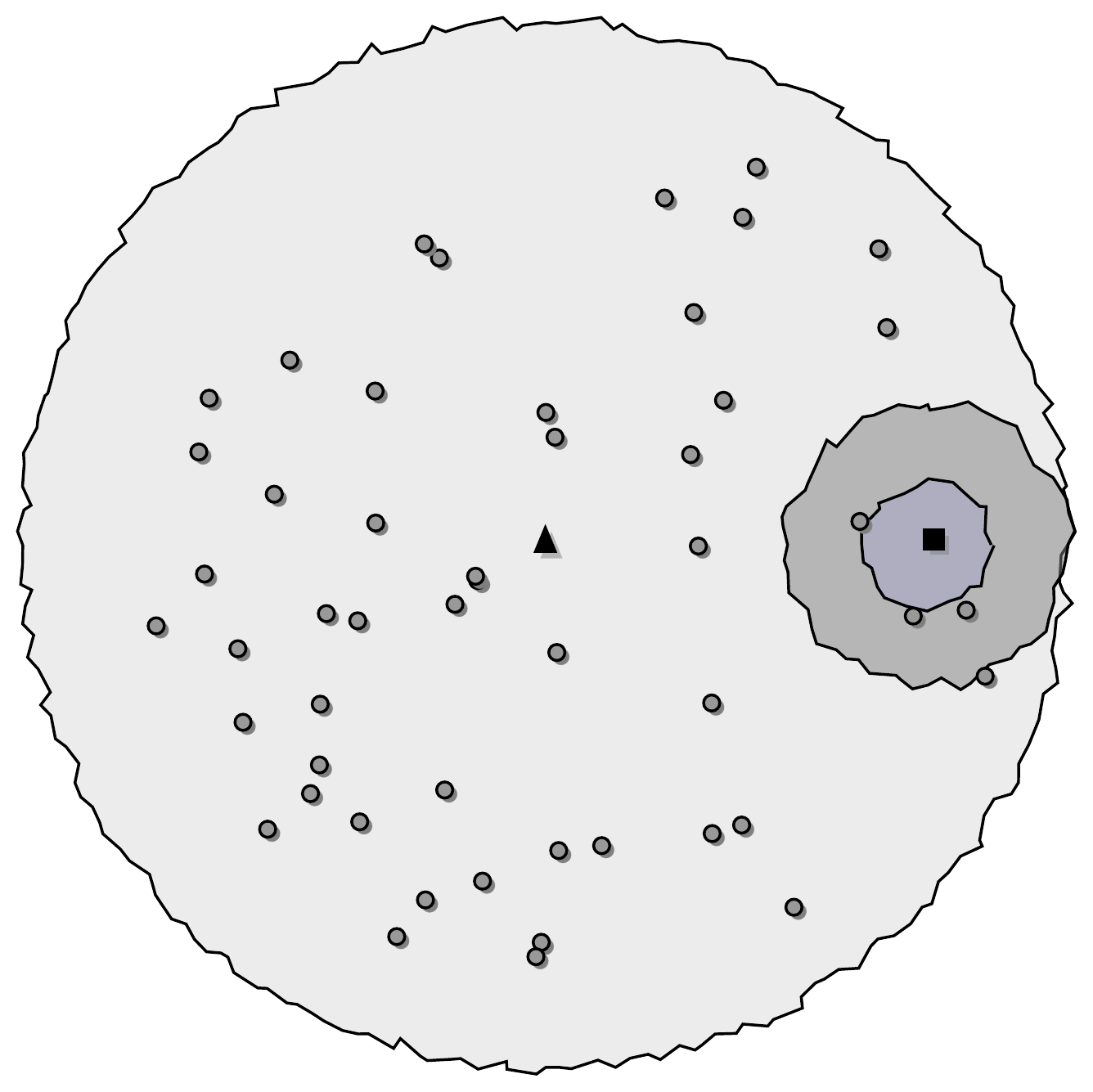}
	\begin{tikzpicture}[overlay]
		%
		\draw[|<->|, thick, line width = 2pt] (6, 7.7) -- node[rotate=90, above
		right=.5em, sloped] {$\mathbf{\mathrm{REB}}$} (6, 6.5);
	\end{tikzpicture}
	\begin{tikzpicture}[overlay, domain=1:8]
		\draw[->, thick] (-1, 8.75) -- (9, 8.75) node[below] {distance};
		\draw[->, thick] (-.1, 7.4) -- (-.1, 18) node[above] {Received power};
		\draw[color=red, xshift=-1cm, yshift=8cm, line width=2pt] plot[id=mbs]
		function{2/log(x)} node[right=0em] {$\mathbf{y^\mathrm{M} \left( d
		\right)}$};
		\draw[color=red, xshift=-2cm, yshift=8.5cm, line width=2pt, dashed]
		plot[id=pbs] function{1/log(-x + 9)} node[right=0em]
		{$\mathbf{y^\mathrm{P} \left( d \right) + \Delta \mathrm{REB}}$};
		\draw[color=red, xshift=-2cm, yshift=8.5cm, line width=2pt]
		plot[id=reb] function{.5/log(-x + 9)} node[right=0em]
		{$\mathbf{y^\mathrm{P} \left( d \right)}$};
	\end{tikzpicture}
	\caption{Illustration of the \ac{REB} concept.
	Circles indicate \acp{MU}, the shaded triangle depicts the umbrella
	\ac{MBS} and the shaded square depicts the target picocell.}
	\label{FIG:REB}
\end{figure}

From \eqref{EQ:REB_HO}, we derive the probability that the tagged \ac{MU}
within the coverage of the umbrella \ac{MBS} is offloaded to the target
\ac{PBS}.
The \ac{LN} approximation in \eqref{EQ:LOGNORMAL_PARAMETERS} is used here to
recover the distribution of the received power at the tagged receiver.
\begin{proposition}
	Consider the observation region $\mathcal{O}$ centered at the tagged
	receiver and the biased cell association as described above; then, the
	probability that the tagged receiver connects to the target \ac{PBS} is
	given by,
	\begin{align}
		\Pr \left[ Y^\mathrm{M} < Y^\mathrm{P} + \delta \right] \simeq
		\sum^{K}_{k=1} {\frac{\omega_k}{2 \sqrt{\pi}} g \left( \eta_k \right)},
		\label{EQ:REB_PROBABILITY}
	\end{align}
	{\noindent where $\delta = \Delta \mathrm{REB} - \Omega$, $\eta_{k}$ is the
	$k^\mathrm{th}$ zero of the Hermite polynomial $H_{K} \left( \eta \right)$
	of degree $K$, $\omega_{k}$ is the corresponding weight of the function $g
	\left( \,\cdot\, \right)$ at the $k^\mathrm{th}$ abscissa and $g \left(
	\eta \right) = 1 + \operatorname{Erf} \left[\frac{-\mu_M + \mu_P + \sqrt{2}
	\eta \sigma_P}{\sqrt{2} \sigma_M}\right]$.}
	\label{PROPOSITION:REB_PROBABILITY}
\end{proposition}
\begin{IEEEproof}
	See Appendix \ref{PROOF:REB_PROBABILITY}.
\end{IEEEproof}
When using the standard procedure in \eqref{EQ:STD_HO}, one needs to make the
substitution $\delta = - \Omega$ in \eqref{EQ:REB_PROBABILITY} to derive the
handover probability.

Fig. \ref{FIG:HO_PROB} shows the handover probability for distinct network
configurations.
To generate this plot, we consider that the tagged receiver is randomly placed
around the umbrella \ac{MBS} in an annular region with inner radius equal to
$25\,\mathrm{m}$ and outer radius of either $250\,\mathrm{m}$ or
$500\,\mathrm{m}$.
Notice that this region actually defines the minimum and maximum distances
between the tagged receiver and the umbrella \ac{MBS}.
In addition, the distance from the tagged receiver to its serving picocell
varies within the set $\left\{ 15, 30, 45 \right\}\mathrm{m}$.
When the \ac{MU} is near to the target picocell, the \ac{REB} does not affect
the handover probability so significantly.
However, the \ac{REB} effect becomes pronounced when the user is located
farther away from the picocell of interest.
For sake of illustration, we consider the tagged user located $45\,\mathrm{m}$
away from the target picocell and bias of $\Delta \mathrm{REB} =
5\,\mathrm{dB}$.
In contrast to the standard approach in \eqref{EQ:STD_HO}, the handover
probability increases from $38\%$ to $54\,\%$ (dashed line with up-triangles).
\begin{figure}[h!]
	\centering
	\includegraphics[width=1.\columnwidth]{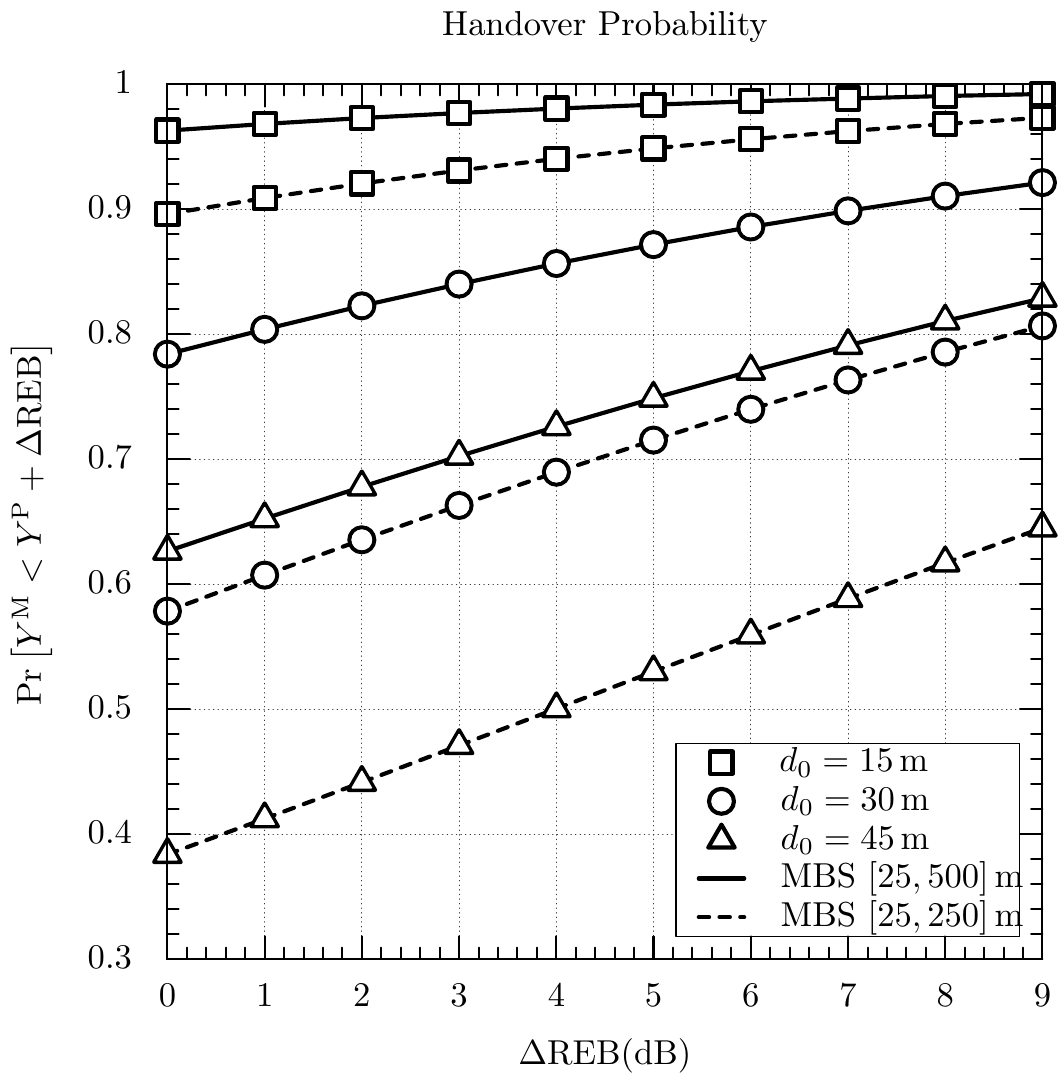}
	\caption{Handover probability as a function of increasing $\Delta
	\mathrm{REB}$ values.}
	\label{FIG:HO_PROB}
\end{figure}

\section{Network Operation}
\label{SEC:NETWORK_OPERATION}
In the coexistence scenarios under study, self-organizing \acp{PBS} employ
distributed strategies to control the cross-tier interference
\cite{ART:PREHOFER-ICM05}.
Hereafter, we describe these solutions and translate their operation to our
mathematical framework so as to identify their impact on the overall system
performance.

\subsection{\acl{ABS}}
\label{SEC:ABS}
By observing Fig. \ref{FIG:REB}, it is clear that within the range expanded
region the received power of the target picocell with \ac{REB} is weaker than
the umbrella \ac{MBS}.
To cope with this problem in such \ac{SON}, the \ac{ABS} is considered as a
baseline strategy to implement interference control.
In fact, \ac{ABS} is a time-domain resource partitioning strategy whereby
\acp{MU} in the expanded region of picocells only transmit within the reserved
slots.
During these reserved slots, the umbrella \ac{MBS} either implements soft
\ac{ABS} by transmitting with less power; or does not transmit at all what
characterizes the zero-power \ac{ABS} \cite{TR:3GPP-TS36.423}.
In Fig. \ref{FIG:ABS}, the aggressor \ac{MBS} does not transmit during the
reserved slots so as to protect the range expanded picocells.
\begin{figure}[h!]
	\centering
	\includegraphics[width=1.\columnwidth]{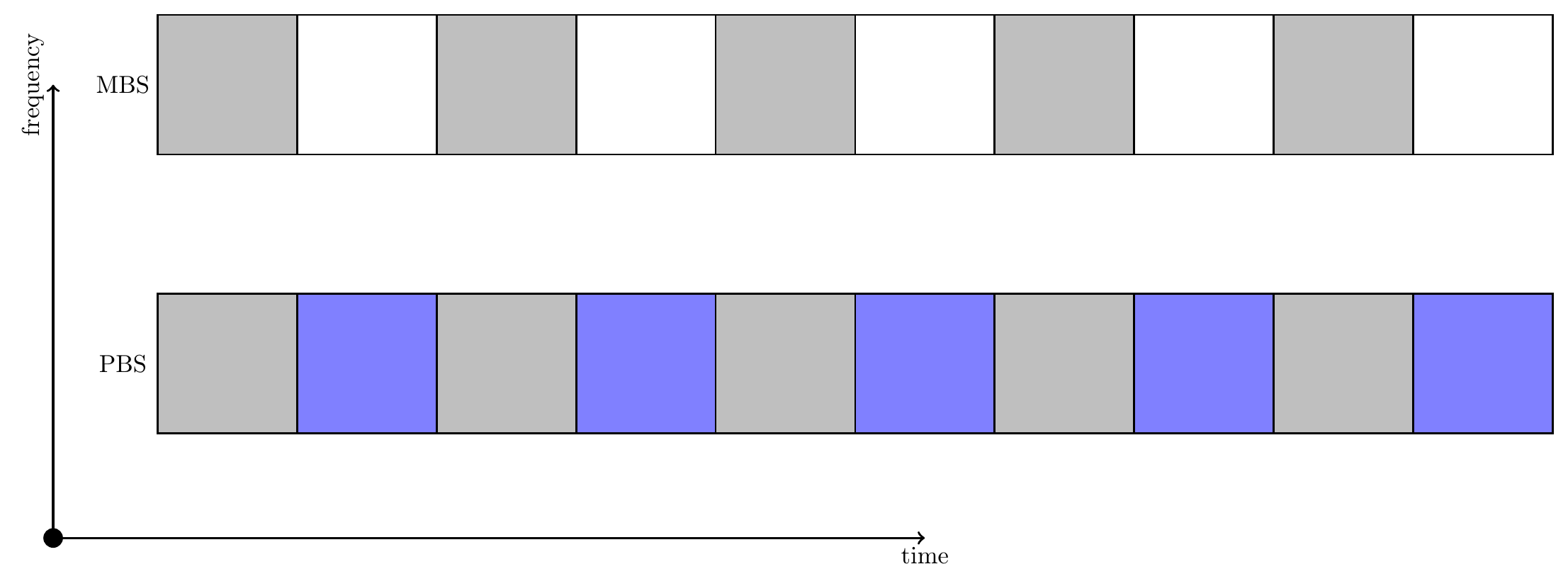}
	\caption{Illustration of the \ac{ABS} strategy with rate of $1/2$.
	The umbrella \ac{MBS} leaves every second subframe (reserved slot) empty so
	that cell edge \acp{MU} which were reassigned to the picocell tier
	experience less interference.}
	\label{FIG:ABS}
\end{figure}
We follow \cite{PROC:OKINO-ICC11} in assuming that only cells with \ac{REB} are
allowed to transmit within the reserved subframes.
The downside of the \ac{ABS} strategy is that non-\ac{REB} cells undergo
capacity loss since reserved subframes are left idle.
Different from the typical approach and depending on the network configuration,
we consider that \ac{ABS} applies to both macro and picocells.
There is no loss of generality in assuming that the umbrella \ac{MBS} reserves
$ 1/2$ of the frame for the \ac{ABS} allocation \cite{ART:PEREZ-JSAC12} when
operating with \ac{REB}.

\subsection{Downlink--\acl{HII}}
\label{SEC:DL_HII}
To avoid the inherent capacity loss of the \ac{ABS} strategy, we also
investigate distributed strategies that rely on the autonomous coordination of
nearby \acp{BS}.
Inspired by the busy tones concept \cite{ART:OMIYI-TWC07} and the interference
mitigation technique in \cite{TR:3GPP-R4-093196}, we consider the utilization
of \ac{DL} channel measurements for the coordination of interfering picocells.
A bitmap indicator which is similar to the \ac{RNTP} indicator in Release $8$
is used to identify dominant interferers \cite{BOOK:SESIA-WILEY09}.
The tagged \ac{MU} identifies potential interferers by monitoring their pilot
signal and reporting the measurements to its serving \ac{BS}.
After acquiring this measurement report, the serving \ac{BS} then coordinates
by exchanging the interference bitmap with the surrounding picocells via the
X$2$ interface.
The updating period of the \ac{DL}-\ac{HII} messages is a configurable
parameter which is comparable to the handover procedure
\cite{ART:BHARUCHA-EURASIP10}.

Within our mathematical framework, the tagged receiver uses the interference
threshold $\rho_{\mathrm{th}}$ to identify potential interferers in its
vicinity.
Since our network operates in \ac{TDD} mode, we assume that the channels for
measurements and data transmissions are fully correlated.
Furthermore, the channel gain between active interferers and the tagged
receiver are assumed to be perfectly estimated by the receiver of interest.

Under the above assumptions, the set of dominant interferers is identified by
the following indicator function,
\begin{equation}
	\mathds{1}_{\widetilde{\Phi}}\left( p_b r^{-\alpha} x \right) =
	\left\{
	\begin{array}{l@{,\,}l}
		1	& \thickspace \text{if} \thickspace p_b r^{-\alpha} x \geq
		\rho_{\mathrm{th}} \\
		0	& \thickspace \text{otherwise},
	\end{array}
	\right.
	\label{EQ:AGGREGATE_CCI_DLHII}
\end{equation}
which defines the first coordination region denoted by $ \mathcal{R}_{1}$, and
where $p_b$ is the transmit power of the reference signal of surrounding
picocells.

In accordance with the formulation of Section \ref{SEC:DEPLOYMENT_MODEL},
picocells within this region constitute a \ac{MPP} which is denoted by
$\displaystyle \widetilde{\Phi}_{1} = \left\{ \left( \varphi, x \right) \in
\widetilde{\Phi}\,|\,p_b r^{-\alpha} x \geq {\rho}_{{th}} \right\}$.
Similarly, picocells in $\mathcal{R}_2$, which are not detected by the \ac{MU}
of interest, form the process $\widetilde{\Phi}_{2} =
\widetilde{\Phi}\backslash\widetilde{\Phi}_{1}$.
Notice that the coordination regions $\mathcal{R}_1$ and $\mathcal{R}_2$ are
disjoint and statistically independent by construction, therefore it follows
immediately from the Superposition theorem \cite{BOOK:KINGMAN-OXFORD93} that
$\widetilde{\Phi} = \widetilde{\Phi}_{1}\cup \widetilde{\Phi}_{2}$.
Fig. \ref{FIG:COORDINATION_REGIONS} illustrates the resulting coordination
regions by following the criterion in \eqref{EQ:AGGREGATE_CCI_DLHII}.
It is worth noticing that after coordinating, the tagged receiver is only
interfered by active transmitters in $\mathcal{R}_2$, since nodes in
$\mathcal{R}_1$ switch to non-conflicting resource allocation.
\begin{figure}[h!]
	\centering
	\includegraphics[width=1.\columnwidth]{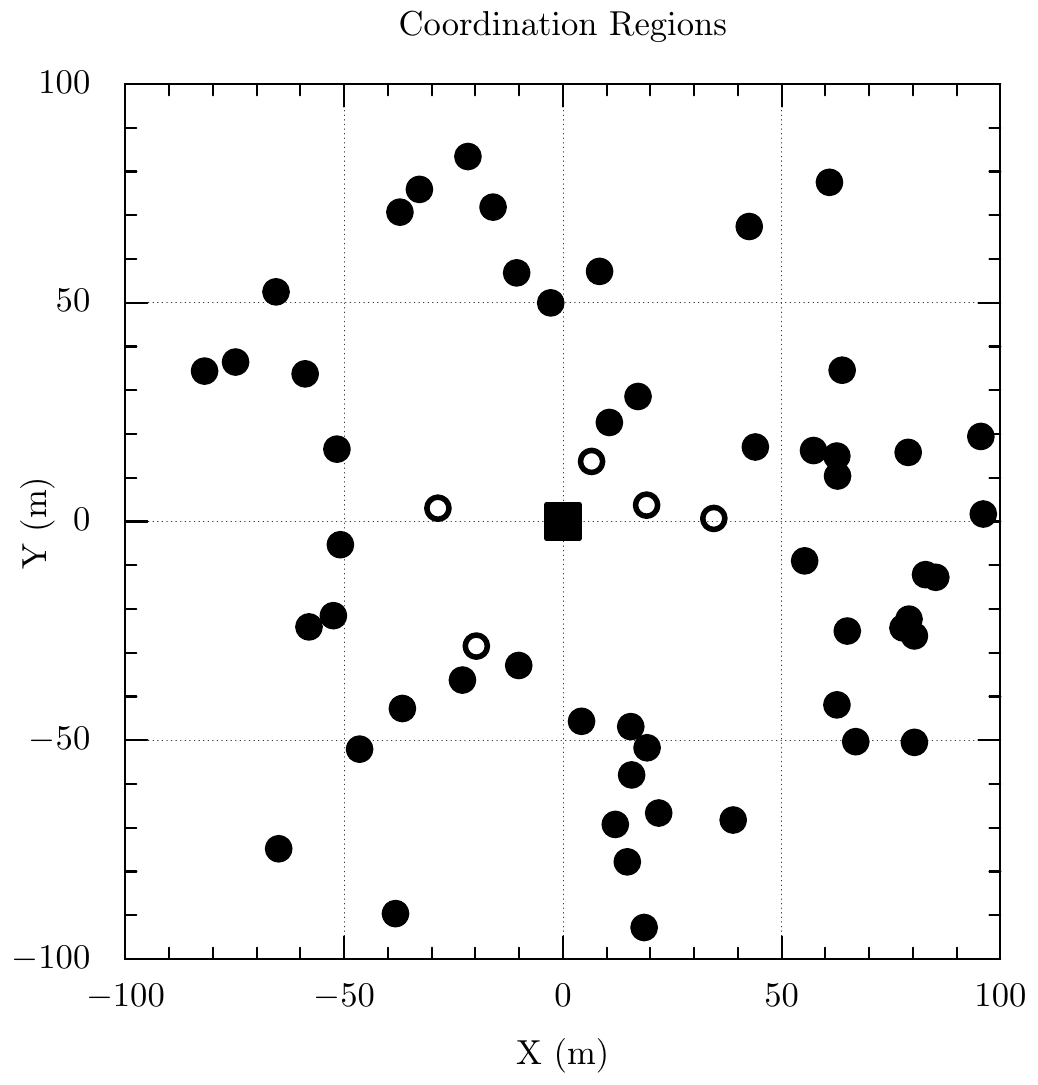}
	\caption{Illustration of the interfering regions.
	Unshaded circles identify dominant interferers within $\mathcal{R}_1$ whose
	received power is above the threshold predefined $\rho_\mathrm{th}$.
	Shaded circles identify remaining interferers within $\mathcal{R}_2$.}
	\label{FIG:COORDINATION_REGIONS}
\end{figure}

\subsection{Virtual \acl{DAS}}
\label{SEC:VDAS}
By following the standard \ac{REB} strategy, a user is served by a \ac{BS}
which does not actually provide the strongest received power which brings about
the side effect of exposing the user of interest to high interference levels.
With the virtual \ac{DAS} strategy, we intend to further exploit the \ac{HII}
bitmap information so that surrounding picocells coordinate over the X$2$
interface and establish virtual \ac{DAS} with random antenna layout.
As discussed in \cite{TR:3GPP-RP111117}, such techniques are actually a work
item in the \ac{3GPP} standardization where large scale remote radio heads are
seen as a promising solution to meet the requirements of release $11$.
Note that with this solution we do not consider \ac{TAS} or precoding and the
received signals of the serving \ac{DAS} coherently add at the user of
interest.
Instead of considering the \ac{REB}, the tagged \ac{MU} uses the aggregate
received power from the serving group which belongs to $\widetilde{\Phi}_1$.
Similar to the coordinated multipoints strategy, this solution requires the
user data to be available at all coordinating picocells which in its turn
requires extra singling exchange and more elaborated backhaul infrastructure
\cite{ART:ZHANG-JSAC10}.
Our strategy is similar to the maximum ratio transmission
\cite{ART:ZHANG-TWC08} by which all antenna elements transmit the same
information and the aggregate received power is 
\begin{align}
	\sum_{\left ( \varphi, x \right ) \in \widetilde{\Phi}_1} \negthickspace
	Y^\mathrm{P} \negthickspace \left( \varphi, x \right)
	\label{eq:mrt}
\end{align}
{\noindent where $Y^\mathrm{P} \negthickspace \left( \varphi, x \right)$ yields
the received power from the picocell at $\varphi$ with shadowed fading $x$.}
\textcolor{black}{It is worth noticing that the random process
$\widetilde{\Phi}_1$ has intensity given by $\Pr \left[ \thickspace p_b
r^{-\alpha} x \geq \rho_{\mathrm{th}} \right] \lambda f_X \left( x \right)$.
}

In what follows, we initially extend the analytical framework presented in
\cite{ART:LIMA-TWC12A} to evaluate how these distributed strategies perform in
heterogeneous networks composed of self-organizing small cells and legacy
macrocells.
Thereafter, the performance of the \ac{SON} is evaluated in terms of \ac{SIR},
outage probability and average spectral efficiency as shown in Section
\ref{SEC:PERFORMANCE_ANALYSIS}.

\section{Interference Model}
\label{SEC:CCI}
Under the assumptions of Section \ref{SEC:SYSTEM_MODEL} and with respect to the
tagged receiver, we now use our analytical framework to derive the probability
distributions of the desired signal and the resulting aggregate \ac{CCI} for
each one of the \acp{ES} described next.
In fact, each \ac{ES} characterizes a particular network configuration, in
which macro and picocells employ distributed strategies to mitigate the
cross-tier interference.
The following list summarizes the evaluation scenarios under consideration.
\begin{itemize}
	\item \ac{ES}$1$: the tagged \ac{MU} connects to the umbrella \ac{MBS} and
		experiences full interference from the underlaid picocell tier.
	\item \ac{ES}$2$: the tagged \ac{MU} connects to the target \ac{PBS} with
		\ac{REB}, but no \ac{ICIC} strategy, such as \ac{ABS}, is carried
		out.
		As a result, the user of interest which dwells in the \ac{RE} region of
		the serving picocell is subject to high interference levels from the
		umbrella macrocell.
	\item \ac{ES}$3$: the tagged \ac{MU} connects to the target \ac{PBS} with
		\ac{REB}, and the umbrella \ac{MBS} implements the \ac{ABS} scheme with
		rate $1/2$.
	\item \ac{ES}$4$: the tagged \ac{MU} connects to the target \ac{PBS}, and
		the surrounding picocells coordinate based on the \ac{DL}-\ac{HII} bitmap
		\cite{ART:BHARUCHA-EURASIP10}, although the umbrella \ac{MBS} still
		interferes.
		This configuration is particularly relevant when the density of small
		cells nearby the tagged receiver is high, or there are multiple tiers
		of interfering small cells, such as femtocells.
	\item \ac{ES}$5$: based on the \ac{DL}-\ac{HII} bitmap, the strongest
		picocells coordinate so as to implement a virtual \ac{DAS}.
		However, the umbrella \ac{MBS} and small cells in $\mathcal{R}_2$ still
		interfere with the user of interest.
		It is assumed that the picocells coordinate through the X$2$ interface.
		However, picocells can also coordinate over the air interface for
		example using the coordination mechanism introduced in
		\cite{ART:LIMA-TWC12A}.
\end{itemize}

\subsection{Received Power from the Umbrella \ac{MBS}}
\label{SEC:MBS_CCI}
In this section, we initially derive the \ac{CF} \cite{BOOK:DALE-WILEY79,
BOOK:ABRAMOWITZ-DOVER03, ART:GHASEMI-JSAC08} of the \ac{RV} which describes the
power received from a random transmitter within $\mathcal{O}$ and thereafter
particularize it to the umbrella \ac{MBS} case.
By considering the communication model of Section \ref{SEC:DEPLOYMENT_MODEL},
we write the \ac{CF} of the power received at the tagged \ac{MU} from a random
transmitter within its observation region as follows.
\begin{proposition}
	Let $Y = R^{-\alpha} X$ be a \ac{RV} describing the power received at the
	tagged receiver from a random transmitter in $\mathcal{O}$ with $R$ varying
	from $R_m$ to $R_M$ and $X$ following the Gamma-\ac{LN} distribution of
	Section \ref{SEC:SYSTEM_MODEL}.
	Then, the \ac{CF} of $Y$ is
	\begin{align}
		\Psi_Y \left ( \omega \right ) = \frac{2}{R_M^2-R_m^2}
		\mathrm{E}_X\left [ \operatorname{R} \left( \omega \right) \right ]
		\hspace{-.25em},
		\label{EQ:INTERFERENCE_COMPONENT_CF}
	\end{align}
	{\noindent where $\operatorname{R} \left( \omega \right) = \int_{R_m}^{R_M}
	{ \exp{\left( j \omega p r^{-\alpha} x \right)} r \mathrm{d}r }$ and
	$\mathrm{E}_X\left [ \,\cdot\, \right ]$ yields the expectation of the
	enclosed expression over the \ac{RV} $X$.}
	\label{PROPOSITION:INTERFERENCE_COMPONENT_CF}
\end{proposition}
\begin{IEEEproof}
	See Appendix \ref{PROOF:INTERFERENCE_COMPONENT_CF}.
\end{IEEEproof}
It is worthy noting that \eqref{EQ:INTERFERENCE_COMPONENT_CF} is a general
formulation which characterizes the distribution of any random transmitter
within the reception range of the tagged receiver including the umbrella
\acp{MBS} and small cells in the underlaid tier.

Thereafter, by taking the $n^{\mathrm{th}}$ derivative of the \ac{CF} as given
in \eqref{EQ:CUMULANT_PARTIAL_DERIVATIVE}, the corresponding cumulant
$\kappa_n$ is obtained.
\begin{proposition}
	Consider the \ac{CF} of the power received from a transmitter randomly
	deployed within the observation region $\mathcal{O}$; then, the
	$n^\mathrm{th}$ cumulant of $Y$ is given by
	\begin{align}
		\kappa_n &= \frac{1}{j^n} \sum_{k=0}^n g^{\left( k \right)} \left(
		\beta_0 \right) \cdot B_{n,k}\left[ \beta_1, \beta_2, \dots,
		\beta_{(n-k+1)} \right],
		\label{EQ:INTERFERENCE_COMPONENT_CUMULANT}
	\end{align}
	{\noindent where $\displaystyle g \left ( u \right ) = \ln {\left ( u
	\right ) }$, $B_{n,k}\left[ \beta_1, \beta_2, \dots, \beta_{(n - k + 1)}
	\right]$ is the partial Bell polynomial \textup{\cite{ART:BELL-TAM34}} and
	$\beta_n = j^n p^n \times \frac{R_m^{2-n \alpha } - R_M^{2-n \alpha }}{n
	\alpha - 2} \operatorname{E}_X \left[ x^n \right]$.}
	\label{PROPOSITION:INTERFERENCE_COMPONENT_CUMULANT}
\end{proposition}
\begin{IEEEproof}
	See Appendix \ref{PROOF:INTERFERENCE_COMPONENT_CUMULANT}.
\end{IEEEproof}


\subsection{Aggregate \ac{CCI} from the Underlaid Tier of Small Cells}
\label{SEC:PBS_CCI}
This scenario represents our default configuration in which the umbrella
\ac{MBS} serves the tagged receiver, whereas the underlaid picocell tier is the
only source of interference.
In order to characterize the distribution of the aggregate \ac{CCI} at the
tagged receiver, we use the cumulant-based framework with the \ac{MPP}
$\widetilde{\Phi}$ \cite{ART:GHASEMI-JSAC08, ART:LIMA-TWC12A}.
By applying Campbell's theorem to \eqref{EQ:MPP}, we derive its characteristic
functional \cite{BOOK:KINGMAN-OXFORD93} as given next.
\begin{proposition}
	Consider the \ac{ES}$1$; then, the $n^\mathrm{th}$ cumulant of the
	aggregate \ac{CCI} perceived by the tagged \ac{MU} within $\mathcal{O}$ and
	with respect to $\widetilde{\Phi}$ is given by,
	\begin{align}
		&\kappa_{n}\negthickspace \left( \widetilde{\Phi} \right) =\frac{2 \pi
		\lambda\,p^{n}}{n \alpha -2} \left(R_m^{2-\alpha
		n}\negthickspace-R_M^{2-\alpha  n}\right)
		\operatorname{E}_X^{n}\negmedspace\left [ 0, \infty \right ].
		\label{EQ:CUMULANT_PBS_FULL_INTERFERENCE}
	\end{align}
	\label{PROPOSITION:CUMULANT_PBS_FULL_INTERFERENCE}
\end{proposition}
\begin{IEEEproof}
	See Appendix \ref{PROOF:CUMULANT_PBS_FULL_INTERFERENCE}.
\end{IEEEproof}
The aggregate \ac{CCI} from the underlaid tier of picocells is computed with
respect to a limited region of the total field of interfering nodes (from $R_m$
to $R_M$).
To account for the neglected interference parcel beyond $R_M$, one needs to
change in ($32$) the upper limit of integration with respect to $r$ until
$\infty$.
For instance, considering $\alpha = 3$, $R_m = 5$m and $R_M = 250$m, the
aggregate interference from the region beyond $R_M = 250$m (towards $\infty$)
represents only $2\%$ of the aggregate interference, whereas for an observation
region within $R_m = 25$m and $R_M = 500$m the neglected region contributes
with $5\%$ of the total interference.

\subsection{Aggregate \ac{CCI} from Multiple Tiers}
The additivity property of cumulants is used to compute the aggregate \ac{CCI}
in heterogeneous scenarios with multiple tiers \cite{BOOK:KENDALL-GRIFFIN45}.
In order to apply this property, the interference components are assumed to be
independent.
\begin{proposition}
	Consider the two tier deployment scenario where an umbrella \ac{MBS} is
	underlaid with self-organizing small cells; then, the $n^\mathrm{th}$
	cumulant of the aggregate \ac{CCI} perceived by the tagged \ac{MU} in
	$\mathcal{O}$ is,
	\begin{align}
		\kappa_{n} = \kappa_{n}^\mathrm{M} + \kappa_{n}^\mathrm{P}.
		\label{EQ:CUMULANT_PBS_PLUS_MBS}
	\end{align}
	\label{THEOREM:CUMULANT_PBS_PLUS_MBS}
\end{proposition}
\begin{proof}
	Since the interference components from both tiers are independent of each
	other, we can use the cumulants additivity property to obtain
	\eqref{EQ:CUMULANT_PBS_PLUS_MBS}.
\end{proof}

\subsection{Aggregate \ac{CCI} with \acl{ICIC}}
\label{SEC:PBS+RE+ABS+DLHII}
As shown in Section \ref{SEC:DL_HII}, picocells use the \ac{DL}-\ac{HII} bitmap
to self-organize into two coordination regions $\mathcal{R}_1$ and
$\mathcal{R}_2$.
Herein, we derive the cumulants of the aggregate \ac{CCI} generated by each
such region with respect to the tagged receiver.
Small cells that are detected by the tagged receiver within $\mathcal{R}_1$
decrease their transmit power by a predefined value, \textit{i.e.}, $p^{\prime}
= p - \Delta p$ so as to reduce their interference towards the user of
interest.
In the following proposition, the cumulants of the dominant interfering
picocells belonging to \ac{MPP} $\widetilde{\Phi}_{1}$ are identified.
\begin{proposition}
	Consider the network operation of Section \ref{SEC:DL_HII}; then, the
	$n^\mathrm{th}$ cumulant of the aggregate \ac{CCI} perceived by the tagged
	\ac{MU} in $\mathcal{O}$ with respect to $\widetilde{\Phi}_{1}$ is written
	as,
	\begin{align}
		\kappa_{n} \negthickspace \left( \widetilde{\Phi}_{1} \right) =
		\negmedspace\frac{2 \pi \lambda\,(p^{\prime})^n}{n\alpha - 2}
		&\bigg\{\hspace{-.5em} \left(R_m^{2-\alpha n}-R_M^{2-\alpha  n}\right)
		\operatorname{E}_X^n{\left [ {\varrho}_M, \infty \right ] } -
		{\varrho}_{th}^{n-\frac{2}{\alpha }}
		\operatorname{E}_X^\frac{2}{\alpha}{\left [ {\varrho}_m, {\varrho}_M
		\right ] } \bigg. \nonumber \\
		& \bigg. + R_m^{2-n \alpha } \operatorname{E}_X^n{\left [ {\varrho}_m,
		{\varrho}_M \right ] } \bigg\}.
		\label{EQ:CUMULANT_R1}
	\end{align}
	\label{PROPOSITION:CUMULANT_R1}
\end{proposition}
\begin{IEEEproof}
	See Appendix \ref{PROOF:CUMULANT_R1}.
\end{IEEEproof}

During the coordination mechanism, the tagged receiver does not detect the
picocells within $\mathcal{R}_2$ which contribute to the aggregate interference
with transmit power $p\,\mathrm{dBm}$.
Therefore, the distribution of the remaining interference is characterized by
the following cumulants.
\begin{proposition}
	Consider the network operation of Section \ref{SEC:DL_HII}; then, the
	$n^\mathrm{th}$ cumulant of the aggregate \ac{CCI} perceived by the tagged
	\ac{MU} in $\mathcal{O}$ with respect to $\widetilde{\Phi}_{2}$ has the
	following form,
	\begin{align}
		\kappa_{n}\negthickspace \left( \widetilde{\Phi}_{2}
		\right)=\negmedspace\frac{2 \pi  \lambda\,p^n}{n\alpha -
		2}&\bigg\{\hspace{-.5em} \left(R_m^{2-\alpha n}-R_M^{2-\alpha n}\right)
		\operatorname{E}_X^n{\left [ -\infty, {\varrho}_m \right ]} +
		{\varrho}_{{th}}^{n-\frac{2}{\alpha} }
		\operatorname{E}_X^\frac{2}{\alpha} {\left [ {\varrho}_m, {\varrho}_M
		\right ]} \bigg. \nonumber \\
		& \bigg. - R_M^{2-n \alpha } \operatorname{E}_X^n{\left [ {\varrho}_m,
		{\varrho}_M \right ]} \bigg\}.
		\label{EQ:CUMULANT_R2}
	\end{align}
	\label{PROPOSITION:CUMULANT_R2}
\end{proposition}
\begin{IEEEproof}
	See Appendix \ref{PROOF:CUMULANT_R2}.
\end{IEEEproof}

\section{Performance Analysis}
\label{SEC:PERFORMANCE_ANALYSIS}
With regard to the tagged receiver, the performance of the evaluation scenarios
is assessed by means of the resulting outage probability and average channel
capacity.
The scenarios under study are interference limited and hence the thermal noise
is negligible in comparison to the resulting \ac{CCI} \cite{ART:WEBER-TC10}.

\subsection{\acs{SIR} and Outage Probability}
\label{SEC:SIR}
The outage probability is given by $\Pr\left [ \Gamma < \gamma_\mathrm{th}
\right ]$ where the \ac{RV} $\Gamma$ represents the \ac{SIR} distribution of
the tagged receiver, and $\gamma_{\mathrm{th}}$ is the corresponding \ac{SIR}
detection threshold.
\begin{theorem}
	Let $V_{0}$ and $V$ be Normal \acp{RV} (in logarithmic scale) representing
	the power received from the desired transmitter and the aggregate \ac{CCI}
	at the tagged receiver, respectively.
	Under the assumption of the shadowed fading with composite Gamma-\ac{LN}
	distribution, the \ac{SIR} at the tagged receiver is 
	\begin{align}
		\Gamma \thicksim \mathsf{Normal} \left( \mu_{V_{0}}-\mu_{V},
		\sigma^{2}_{V_{0}} + \sigma^{2}_{V} \right),
		\label{EQ:SIR}
	\end{align}
	and the outage probability is given by 
	\begin{align}
		\Pr\left [ \Gamma < \gamma_{\mathrm{th}} \right ] = \text{Q}\left[
		\left( \mu_{\Gamma} - \gamma_{\mathrm{th}} \right)/\sigma_{\Gamma}
		\right],
		\label{EQ:OUTAGE_PROBABILITY}
	\end{align}
	where $ \mu_{\Gamma}=\mu_{V_{0}}-\mu_{V}$ and $
	\sigma_{\Gamma}=\sqrt{\sigma^{2}_{V_{0}} + \sigma^{2}_{V}}$.
	\label{THEOREM:OUTAGE_PROBABILITY}
	\end{theorem}
\begin{IEEEproof}
	The \ac{SIR} distribution is given by the quotient of two independent
	\ac{LN} \acp{RV}, namely, $e^{V_0}$ which is the received power from the
	target transmitter, and $e^{V}$ which is an equivalent \ac{LN} \ac{RV}
	approximating the aggregate \ac{CCI} at the tagged receiver.
	Hence, the multiplicative reproductive property of \ac{LN} \acp{RV} is
	applied to obtain the \ac{SIR} distribution \cite{BOOK:DALE-WILEY79}.
\end{IEEEproof}

\subsection{Average Spectral Efficiency}
\label{SEC:AVERAGE_SPECTRAL_EFFICIENCY}
We evaluate how the two-tier coexistence scenarios perform in terms of the
location-dependent average channel capacity of the tagged receiver
\cite{ART:LEE-TVT90}.
By using the analytical framework previously established, and assuming that all
users are allocated on the same bandwidth $W$, we initially recover the
\ac{SIR} distribution of the tagged receiver, and then compute the
corresponding capacity.
%
%
\begin{theorem}
	Under the assumption of the shadowed fading channel regime, the average
	channel capacity of the tagged receiver is given as,
	\begin{align}
		\bar{C} \simeq W \sum_{k=1}^{K}{\frac{\omega_k}{\sqrt{\pi}} \log_2
		{\left [ 1+ \exp{\left ( \frac{\eta_k \sqrt{2} \sigma + \mu}{\xi}
		\right)}\right ]}}.
		\label{EQ:GAUSS_HERMITE_APPROX}
		\end{align}
	\label{THEOREM:ASE}
\end{theorem}
\begin{IEEEproof}
	To compute the location-dependent average channel capacity, 
	\begin{align}
		\bar{C} = W \int\limits_{0}^{\infty} {\log_{2} {\left ( 1 + \gamma
		\right)}}f_{\Gamma}\left ( \gamma \right )\text{d}\gamma,
		\label{EQ:AVERAGE_CHANNEL_CAPACITY}
	\end{align}
	we use the \ac{PDF} of the \ac{SIR} with respect to the tagged receiver,
	which is indicated by $f_{\Gamma}\left ( \gamma \right )$.	
	The Gauss-Hermite quadrature \cite{BOOK:ABRAMOWITZ-DOVER03} with the
	substitution $\eta=(\xi \ln{\gamma} - \mu)/\sqrt{2} \sigma$ are used to
	obtain \eqref{EQ:GAUSS_HERMITE_APPROX}.
\end{IEEEproof}
\textcolor{black}{As discussed in \cite{ART:ALOUINI-TVT99, ART:MUKHERJEE-JSAC12}
the aggregate interference perceived by the tagged receiver has non-Gaussian
nature, and the Shannon formula is used as a lower bound for the ergodic
rate.}

\section{Numerical Results}
By using the analytical framework of Section \ref{SEC:ANALYTICAL_FRAMEWORK},
the distributions of the received power, aggregate \ac{CCI} and the resulting
\ac{SIR} are calculated for each one the \acp{ES}.
The outage probability and average channel capacity are also used to evaluate
how the system performs with biased cell association and interference
coordination techniques.

Fig. \ref{FIG:PRX_CDF_PBS} compares the \ac{CDF} of the picocell received power
at the tagged receiver ($Y^\mathrm{P}$) from Monte Carlo simulations with those
obtained with the \ac{LN} approximation.
In this example, the annular observation region is defined by $R_m =
5\,\mathrm{m}$ and $R_M = 75\,\mathrm{m}$ and picocells operate with a fixed
power level of $30\,\mathrm{dBm}$.
The radio channel is affected by path loss with exponent $\alpha = 3$, \ac{LN}
shadowing with standard deviation $\sigma = 6,\, 8,\, 10$ and
$12\,\mathrm{dB}$, and Nakagami fading with shape parameter $m = 16$ which
corresponds to a Rician channel with parameter $K = 14.8\,\mathrm{dB}$.
The proposed framework approximates well the power received by the tagged
receiver from a random picocell within its observation region for varying
number of interfering scenarios.
\begin{figure}[!b]
	\centering
	\includegraphics[width=1.\columnwidth]{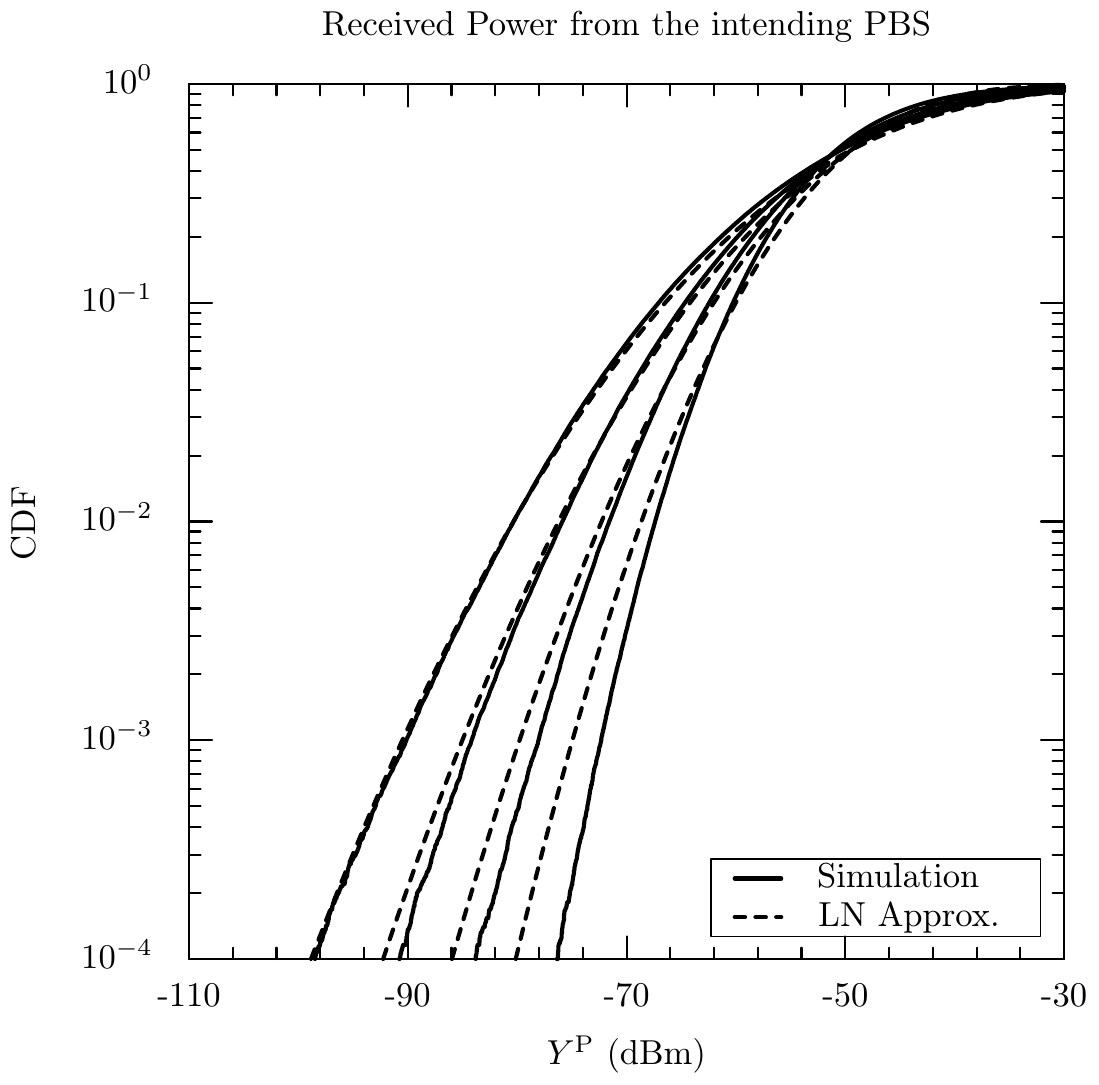}	
	\begin{tikzpicture}[overlay]
		\draw[-stealth', black, line width = 1pt] (3-.33, 9+.33) to (0+.33,
		12-.33);
		\draw (3.5, 9) node[anchor = center] {\small $\sigma = 6, 8, 10,
		12\,\mathrm{dB}$};
	\end{tikzpicture}
	\caption{\ac{CDF} of the power received by the tagged receiver from a
	random picocell within its observation region under shadowed fading with
	$\sigma = \left\{ 6, 8, 10, 12 \right\}\,\mathrm{dB}$ and shape parameter
	$m=16$ (corresponds to a Rician factor $K=14.8\,\mathrm{dB}$).
	\acp{PBS} transmit with constant power equal to $30\,\mathrm{dBm}$.}
	\label{FIG:PRX_CDF_PBS}
\end{figure}
Similarly, Fig. \ref{FIG:PRX_CDF_MBS} shows the \ac{CDF} of the distribution of
the power received from the umbrella macrocell at the tagged receiver
($Y^\mathrm{M}$).
The umbrella \ac{MBS} transmits at $43\,\mathrm{dBm}$.
\begin{figure}[!b]
	\centering
	\includegraphics[width=1.\columnwidth]{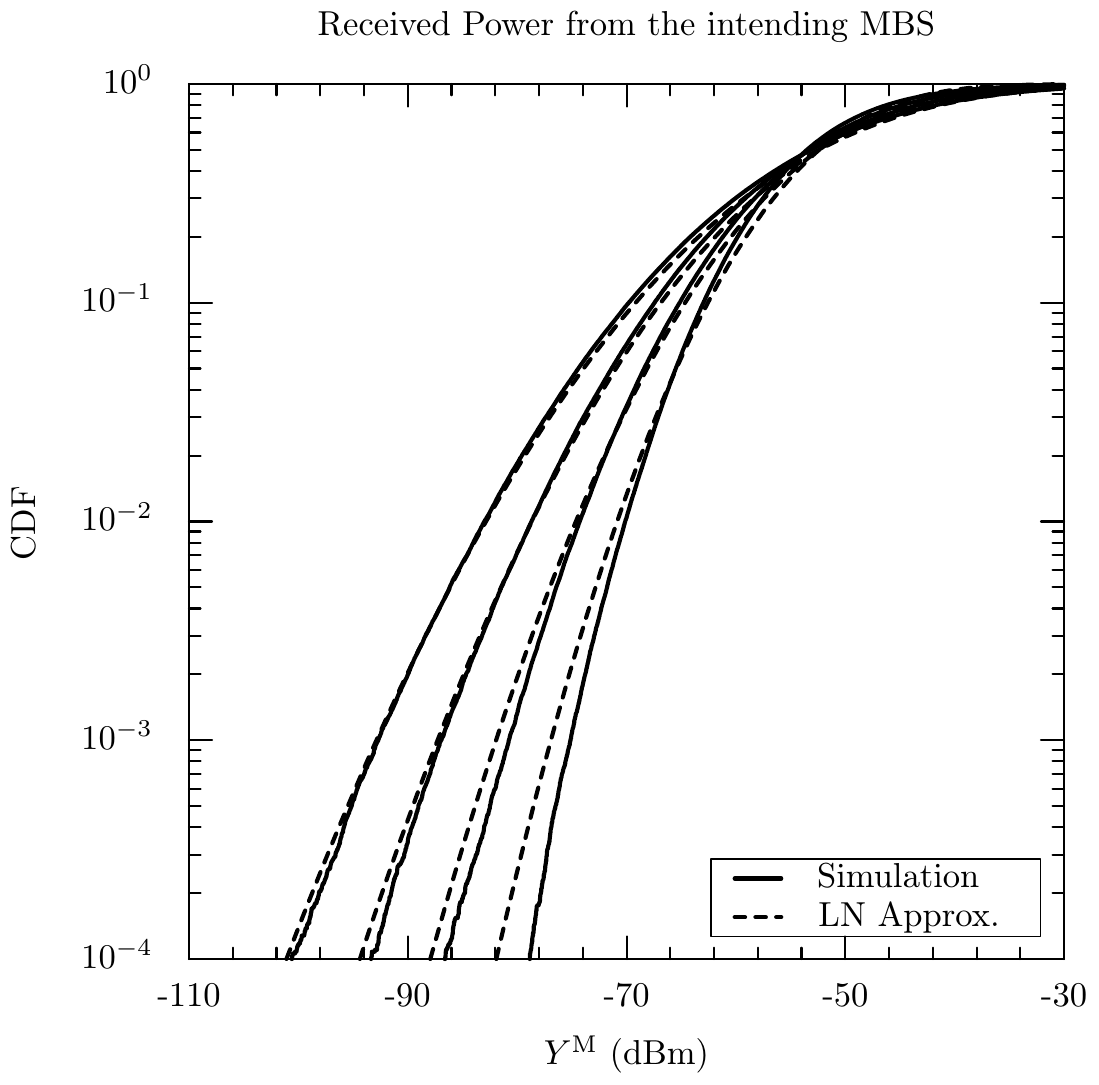}	
	\begin{tikzpicture}[overlay]
		\draw[-stealth', black, line width = 1pt] (3-3*.33, 9+3*.33) to
		(0-1*.33, 12+1*.33);
		\draw (3, 9.5) node[anchor = center] {\small $\sigma = 6, 8, 10,
		12\,\mathrm{dB}$};
	\end{tikzpicture}
	\caption{\ac{CDF} of the power received by the tagged receiver from a
	random transmitter within its observation region under shadowed fading with
	$\sigma = \left\{ 6, 8, 10, 12 \right\}\,\mathrm{dB}$ and shape parameter
	$m=16$ (corresponds to a Rician factor $K=14.8\,\mathrm{dB}$).
	The umbrella \acp{MBS} transmit with constant power equal to
	$43\,\mathrm{dBm}$, respectively.}
	\label{FIG:PRX_CDF_MBS}
\end{figure}
As can be seen, our \ac{LN} model matches well the simulation results in the
macrocell configuration.
In addition, the \ac{LN} approximation is tighter for larger $\sigma$ (standard
deviation), when the resulting shadowing dominates the variation of the
received power distribution.

Fig. \ref{FIG:CCDF_AGGREGATE_CCI} compares the \ac{CCDF} of the aggregate
\ac{CCI} from Monte Carlo simulations with those from the proposed \ac{LN}
approximation.
Our approximation matches well with the simulation results for the evaluation
scenarios under study.
An annular observation region with $R_m = 25\,\mathrm{m}$ and $R_M =
250\,\mathrm{m}$ is considered.
Picocells operate with a fixed power level of $30\,\mathrm{dBm}$ and constitute
a Poisson field of interferers with intensity $\lambda =
10^{-5}\,\mathrm{\ac{PBS}}/\mathrm{m}^2$ (about $2$ picocells on average).
By comparing the scenario where \acp{PBS} are the only source of interference
with that in which \acp{PBS} and the umbrella \ac{MBS} jointly interferer, it
is possible to identify the harmful impact of the macrocell component at the
tagged receiver (about $12\,\mathrm{dB}$).
With that effect in mind, the benefits of using \ac{ABS} to avoid the
interference from the umbrella macrocell altogether becomes evident.
Afterwards, we increase the density of picocells to $\lambda =
10^{-4}\,\mathrm{\ac{PBS}}/\mathrm{m}^2$ (about $20$ cells on average) and
allow surrounding picocells to use the downlink \ac{HII} in order to coordinate
with the serving \ac{BS}.
As a result, the interference experienced by the tagged receiver is further
reduced.
An interesting observation is that depending on the density of picocells and
their relative distance to the tagged receiver, the underlay picocell tier
dominates the aggregate interference.
\begin{figure}[h!]
	\centering
	\includegraphics[width=1.\columnwidth]{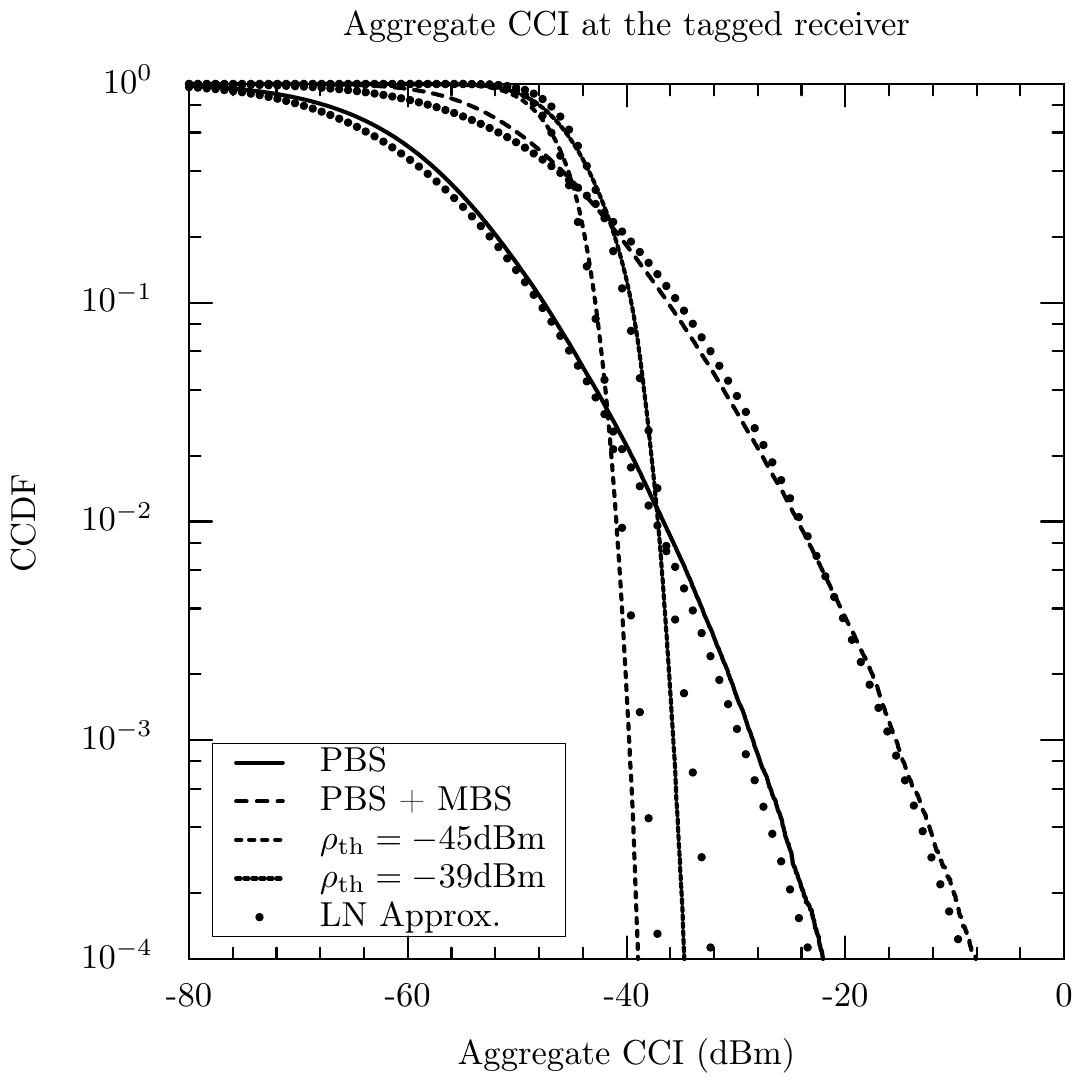}
	\caption{\ac{CCDF} of the aggregate \ac{CCI} $Z$ at the tagged receiver for
	the interference scenarios of Section \ref{SEC:CCI}.}
	\label{FIG:CCDF_AGGREGATE_CCI}
\end{figure}

Fig. \ref{FIG:OUTAGE_PER_DENSITY} shows the outage probability for distinct
evaluation scenarios and increasing density of \acp{PBS}.
The expressions \eqref{EQ:SIR} and \eqref{EQ:OUTAGE_PROBABILITY} in Theorem
\ref{THEOREM:OUTAGE_PROBABILITY} are used to generate the numerical results
shown in this figure.
By comparing the outage probability of the tagged receiver in \ac{ES}$2$ and
\ac{ES}$3$, one observes a performance improvement by avoiding the interference
from the umbrella \acp{MBS}, which is the dominant interferer.
However, the underlaid tier of picocells dominates the resulting interference
as the density of picocells increases -- The \ac{ABS} gains are not so evident
for a density $\lambda$ higher than
$7\times10^{-5}\,\mathrm{\ac{PBS}}/\mathrm{m}^2$.
Hence, the coordination mechanisms are considered in this work to further
reduce the interference levels at the tagged receiver.
When interfering picocells coordinate their transmissions by fulfilling the
criterion $\mathds{1}_{\widetilde{\Phi}}\left( p_b r^{-\alpha} x \right)$ in
scenarios \ac{ES}$4$ and \ac{ES}$5$, the network operation outperforms the
standard configuration wherein \acp{BS} do not self-organize.
In fact, when nodes coordinate following the criterion given in
\eqref{EQ:AGGREGATE_CCI_DLHII}, the tagged receiver experiences much better
link quality since less interferers are active in its reserved subframes.
\begin{figure}[h!]
	\centering
	\includegraphics[width=1.\columnwidth]{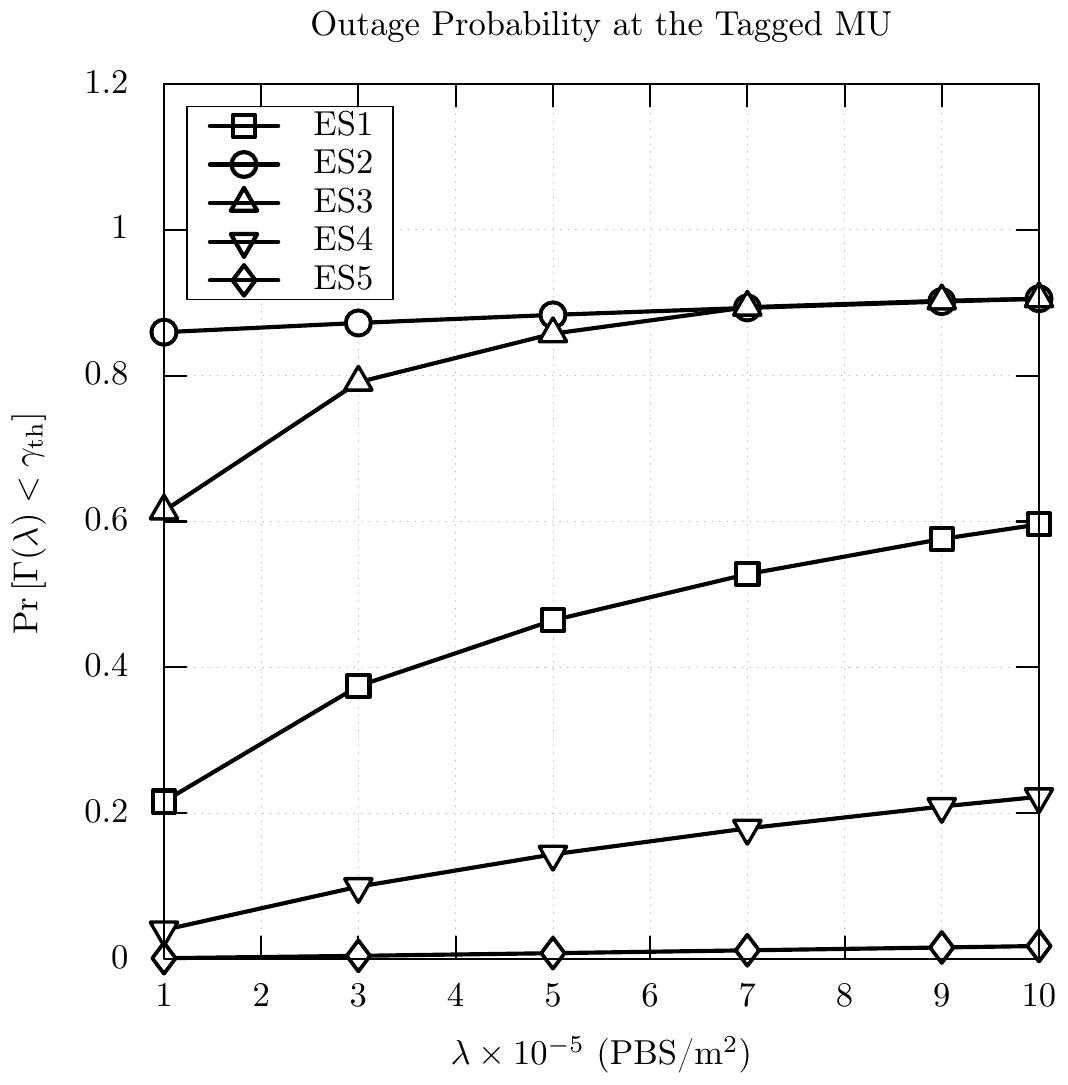}
	\caption{Outage probability at the tagged receiver for increasing density
	of interfering picocells.}
	\label{FIG:OUTAGE_PER_DENSITY}
\end{figure}


In Fig. \ref{FIG:LSE_PER_DENSITY}, we use \eqref{EQ:GAUSS_HERMITE_APPROX} (see
Theorem \ref{THEOREM:ASE}) to compute the average channel capacity of the
tagged link for an increasing density of interfering picocells.
The performance of the tagged receiver is severely degraded by the umbrella
\ac{MBS} which corroborates our previous outage probability results.
By employing interference avoidance techniques in scenarios \ac{ES}$4$ and
\ac{ES}$5$, the channel capacity of the tagged receiver link improves
significantly.
In addition, the tagged receiver benefits mostly from the coordination of
surrounding picocells by means of the \ac{DL}--\ac{HII} which corresponds to
\ac{ES}$4$ and \ac{ES}$5$.
For instance, the tagged receiver attains at most $1\mathrm{bps}/\mathrm{Hz}$
in \ac{ES}$3$, while an average channel capacity of about
$2.5\mathrm{bps}/\mathrm{Hz}$ is achieved in \ac{ES}$5$.
\begin{figure}[h!]
	\centering
	\includegraphics[width=1.\columnwidth]{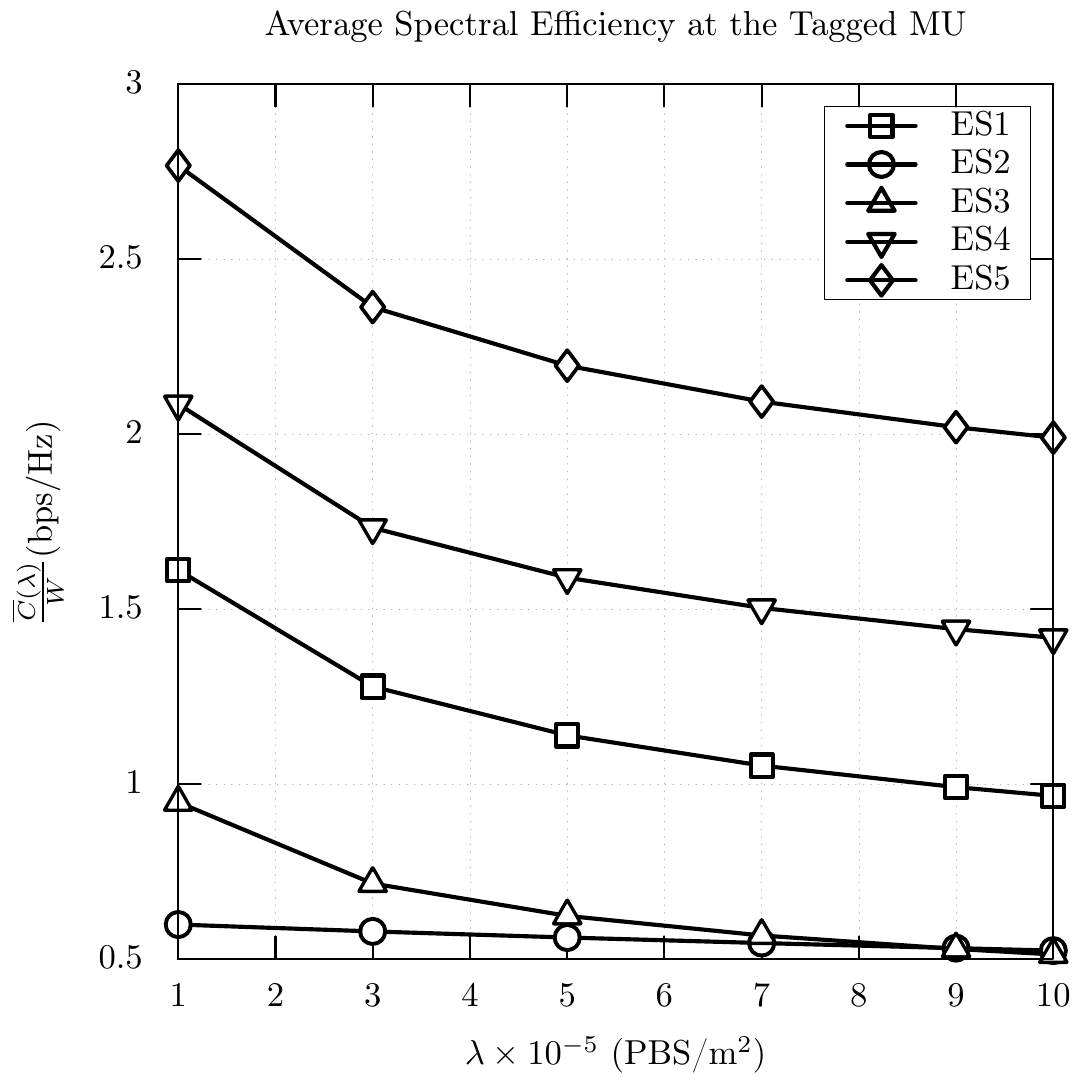}
	\caption{Average channel capacity at the tagged receiver for increasing
	density of interfering picocells.}
	\label{FIG:LSE_PER_DENSITY}
\end{figure}

\section{Conclusions and Final Remarks}
\label{SEC:CONCLUSIONS}
In this paper, we investigate the problem of co-channel interference in
heterogeneous networks composed of self-organizing small cells and legacy
macrocells.
An analytical framework which resorts to stochastic geometry and higher-order
statistics through the concept of cumulants is introduced in order to
characterize network dynamics and channel variations.
We use this framework to recover the distribution of the \ac{CCI} and to
evaluate the system performance in terms of outage probability and average
spectral efficiency of the tagged link.
For the scenarios under study, results show that our analytical model matches
well with numerical results obtained using Monte Carlo simulations.
Aiming to reduce the co-channel interference generated at the underlaid tier,
picocells coordinate their transmissions using the \ac{DL}-\ac{HII} incurring
minimum overhead.
Finally, by employing the concept of virtual \acp{DAS}, the user of interest
not only benefits from reduced interference, but also from the maximum ratio
transmission among the serving picocells.

\begin{appendices}
	\section{Proof of Proposition \ref{PROPOSITION:REB_PROBABILITY}}
	\label{PROOF:REB_PROBABILITY}
	From \eqref{EQ:LOGNORMAL_PARAMETERS}, we know that $Y^\mathrm{M}$ and
	$Y^\mathrm{P}$ follow \ac{LN} distribution with parameters $\left(
	\mu_{\mathrm{M}}, \sigma_{\mathrm{M}} \right)$ and $\left(
	\mu_{\mathrm{P}}, \sigma_{\mathrm{P}} \right)$, respectively.
	Thus, the handover probability is given by
	\begin{align}
		\Pr \left[ Y^\mathrm{M} < Y^\mathrm{P} + \delta \right] =
		\int\limits_0^\infty\int\limits_0^{y^P + c} {f_{Y^M}\left ( y^M \right
		) f_{Y^P} \left ( y^P \right ) \mathrm{d}y^M \mathrm{d}y^P}.
		\label{EQ:HO_PROBABILITY_1}
	\end{align}
	{\noindent where $f_{Y^M} \left ( y^M \right )$ and $f_{Y^P} \left ( y^P
	\right )$ yield the probability density function of the umbrella \ac{MBS}
	and target picocell, respectively.}
	After evaluating the inner-most integral, we obtain
	\begin{align}
		\Pr \left[ Y^\mathrm{M} < Y^\mathrm{P} + \delta \right] =
		\int\limits_0^\infty\frac{1}{2} \operatorname{Erfc}\left[\frac{\mu_M -
		\log\left( c + y^P \right)}{\sqrt{2} \sigma_M}\right] f_{Y^P} \left (
		y^P \right )\mathrm{d}y^P
		\label{EQ:HO_PROBABILITY_2}
	\end{align}
	After making the change of variate $\eta = \frac{-\mu_P + \log\left( y^P
	\right)}{\sqrt{2} \sigma_P}$ in \eqref{EQ:HO_PROBABILITY_2}, we obtain
	\begin{align}
		\Pr \left[ Y^\mathrm{M} < Y^\mathrm{P} + \delta \right] =
		\int\limits_{-\infty}^\infty {\frac{e^{-\eta^2} \operatorname{Erfc}
		\left[ \frac{\mu_M - \log \left( c + e^{\mu_P + \sqrt{2} \eta \sigma_P}
		\right)}{\sqrt{2} \sigma_M} \right]}{2 \sqrt{\pi}}} \mathrm{d} \eta
		\label{EQ:HO_PROBABILITY_3}
	\end{align}
	To evaluate $\Pr \left[ Y^\mathrm{M} < Y^\mathrm{P} + \delta \right]$ in
	\eqref{EQ:HO_PROBABILITY_3}, we then use the Gauss-Hermite quadrature
	\cite{BOOK:ABRAMOWITZ-DOVER03},
	\begin{align}
		\int\limits_{-\infty}^{+\infty} e^{-\eta^2} f(\eta)\,d \eta =
		\sum^{K}_{k=1} {\omega_{k} f \left( \eta_{k} \right) + R_K},
		\label{EQ:GAUSS_HERMITE_RV}
	\end{align}
	{\noindent where $\eta_{k}$ is the $k^\mathrm{th}$ zero of the Hermite
	polynomial $H_{K} \left( \eta \right)$ of degree $K$, $\omega_{k}$ is the
	corresponding weight of the function $f \left( \,\cdot\, \right)$ at the
	$k^\mathrm{th}$ abscissa, and $R_{K}$ is the remainder value.}
	Finally, we obtain \eqref{EQ:REB_PROBABILITY} by performing the
	substitutions indicated above.

	\section{Proof of Proposition \ref{PROPOSITION:INTERFERENCE_COMPONENT_CF}}
	\label{PROOF:INTERFERENCE_COMPONENT_CF}
	From \eqref{EQ:INTERFERENCE_COMPONENT}, the \ac{CF} of the representative
	interference component of a random transmitter within the observation
	region $\mathcal{O}$ is written as,
	\begin{align}
		\Psi_{Y} \left ( \omega \right ) &= \mathrm{E}\left [ e^{j \omega Y}
		\right ] \nonumber \\
		&= \int\limits_0^\infty { \int\limits_{R_m}^{R_M} { e^{j \omega p
		r^{-\alpha} x} f_{R,X} \left ( r, x \right ) \mathrm{d}r\, \mathrm{d}x
		} }.
		\label{EQ:INTERFERENCE_COMPONENT_CF_2}
	\end{align}
	{\noindent where $f_{R,X} \left (r, x \right )$ is the joint density
	function of the separation distance between interferers and the tagged
	receiver, and the shadowed fading.}
	Recalling that the finite field of interferers is within an observation
	region which is delimited by $R_m$ and $R_M$, the \ac{PDF} of the distances
	from random points uniformly scattered within $\mathcal{O}$ to the tagged
	receiver is,
	\begin{align}
		f_R \left ( r \right ) = \frac{2 r}{R_M^2 - R_m^2}.
		\label{EQ:INTERFERENCE_COMPONENT_CF_3}
	\end{align}
	By substituting \eqref{EQ:INTERFERENCE_COMPONENT_CF_3} in
	\eqref{EQ:INTERFERENCE_COMPONENT_CF_2}, we obtain
	\begin{align}
		\Psi_Z \left ( \omega \right ) = \frac{2}{R_M^2 - R_m^2}
		\int\limits_{0}^\infty { \int\limits_{R_m}^{R_M} { \exp{\left( j \omega
		p r^{-\alpha} x \right)}f_X \left ( x \right) r \mathrm{d}r\,
		\mathrm{d}x } }.
		\label{EQ:INTERFERENCE_COMPONENT_CF_4}
	\end{align}
	After manipulating the above expression by performing substitutions and
	simplifications as indicated in Proposition
	\ref{PROPOSITION:INTERFERENCE_COMPONENT_CF}, we obtain
	\eqref{EQ:INTERFERENCE_COMPONENT_CF}.

	\section{Proof of Proposition \ref{PROPOSITION:INTERFERENCE_COMPONENT_CUMULANT}}
	\label{PROOF:INTERFERENCE_COMPONENT_CUMULANT}
	Consider the auxiliary functions $\displaystyle f \left( \omega \right) =
	\int_0^\infty {\int_{R_m}^{R_M} { \exp{ \left( j \omega p r^{-\alpha} x
	\right) } f_X\left ( x \right ) r \mathrm{d}r \mathrm{d}x} }$ and
	$\displaystyle \left( g \circ f \right) \left ( \omega \right) = \ln {
	\left [ f \left ( \omega \right ) \right ] }$.
	Now, using the Fa\`{a} di Bruno's formula \cite{BOOK:ABRAMOWITZ-DOVER03}
	which generalizes the chain rule to compute higher order derivatives of the
	composition of two functions $\left( g \circ f \right) \left( \omega
	\right)$, we have
	\begin{align}
		\frac{\partial^n}{\partial \omega^n} \left( g \circ f \right) \left (
		\omega \right ) &= \sum_{i=0}^n g^{(i)}\left [ f\left ( \omega \right )
		\right] \cdot B_{n,i}\left[ f'(\omega), f''(\omega), \dots,
		f^{(n-i+1)}(\omega)\right],
		\label{EQ:FAA_DI_BRUNOS_FORMULA_1}
	\end{align}
	{\noindent where $B_{n,i}\left[ f'( 0), f''( 0 ), \dots, f^{(n-i+1)}(
	0)\right]$ is the partial Bell polynomial \cite{ART:BELL-TAM34}.}
	After evaluating \eqref{EQ:FAA_DI_BRUNOS_FORMULA_1} at $\omega = 0$ and
	using the definition of cumulants from
	\eqref{EQ:CUMULANT_PARTIAL_DERIVATIVE}, we obtain the following result
	\begin{align}
		\kappa_n  &= \frac{1}{j^n} \sum_{i=0}^n g^{(i)}\left [ f( 0 ) \right]
		\cdot B_{n,i}\left[ f'( 0 ), f''( 0 ), \dots, f^{(n-i+1)}( 0)\right].
		\label{EQ:INTERFERENCE_COMPONENT_CUMULANT2}
	\end{align}
	The derivatives of the auxiliary function $f \left( \omega \right)$ at zero
	are given by,
	\begin{align}
		\beta_n &= \left. \frac{\partial^n f \left ( \omega \right )}{\partial
		\omega^n} \right]_{w = 0} \nonumber \\
		&= j^n p^n \int\limits_0^\infty {x^n f_X \left ( x \right )
		\mathrm{d}x} \int\limits_{R_m}^{R_M} {r^{1 - n\alpha} \mathrm{d}r}.
		\label{EQ:INTERFERENCE_COMPONENT_CUMULANT3}
	\end{align}
	By substituting \eqref{EQ:INTERFERENCE_COMPONENT_CUMULANT3} into
	\eqref{EQ:INTERFERENCE_COMPONENT_CUMULANT2}, the final expression for the
	$n^{\mathrm{th}}$ cumulant of the aggregate \ac{CCI} in
	\eqref{EQ:INTERFERENCE_COMPONENT_CUMULANT} results.
	%
	%
	\section{Proof of Proposition \ref{PROPOSITION:CUMULANT_PBS_FULL_INTERFERENCE}}
	\label{PROOF:CUMULANT_PBS_FULL_INTERFERENCE}
	We start from \eqref{EQ:CHARACTERISTIC_FUNCTION} and apply Campbell's
	theorem \cite{BOOK:KINGMAN-OXFORD93, BOOK:STOYAN-WILEY95} to derive the
	\ac{CF} of the aggregate \ac{CCI} perceived by the tagged \ac{MU} as
	\begin{align}
		\Psi_Z\left ( \omega \right ) = \exp\Bigg\{ 2 \pi
		\int\limits_{0}^\infty \int\limits_{R_m}^{R_M}\negthickspace\Big[
		{}\exp\left ( j w p r^{-\alpha} x \right ) -1\Big] \lambda
		f_X\negthinspace\left ( x \right ) r \text{d}r \text{d}x \Bigg\}.
		\label{EQ:CF_FULL_INTERFERENCE}
	\end{align}
	By substituting \eqref{EQ:CF_FULL_INTERFERENCE} in
	\eqref{EQ:CUMULANT_PARTIAL_DERIVATIVE}, and after integrating with respect
	to $r$, we write the $n^\text{th}$ cumulant as
	\begin{align}
		%
		%
		\kappa_{n}\negthickspace \left( \widetilde{\Phi} \right) = \frac{2 \pi
		\lambda p^n}{n \alpha -2} \left(R_m^{2-\alpha n}-R_M^{2-\alpha
		n}\right) \int\limits_{0}^{\infty} {x^n f_X\left ( x \right )}
		\text{d}x.
		%
		%
		%
		\label{EQ:NTH_CUMULANT_FULL_INTERFERENCE_PROOF}
	\end{align}
	Recalling that $\operatorname{E}_X^n\negthinspace\left [ 0, \infty \right ]
	= \int_{0}^{\infty} {x^n f_X\left ( x \right )} \text{d}x$, we turn our
	attention to the case where transmissions are affected by the shadowed
	fading, and so from Section \ref{SEC:PHY_MODEL},
	$\operatorname{E}_X^n\negthinspace\left [ 0, \infty \right ] = e^{n \mu +
	\frac{1}{2} n^2 \sigma^2}$ which gives
	\eqref{EQ:CUMULANT_PBS_FULL_INTERFERENCE}.
	\section{Proof of Proposition \ref{PROPOSITION:CUMULANT_R1}}
	\label{PROOF:CUMULANT_R1}
	By using the indicator function in \eqref{EQ:AGGREGATE_CCI_DLHII}, we write
	the \ac{CF} of the aggregate \ac{CCI} for the $\mathcal{R}_1$ as,
	\begin{align}
	   \Psi_{Z_1} \left ( \omega \right )=\exp\Bigg\{ 2 \pi
	   \int\limits_{0}^\infty \int\limits_{R_m}^{R_M}\negthickspace\Big[
	   \exp\left ( j w p^{\prime} r^{-\alpha} x \right) - 1\Big]
	   \lambda f_X\negthinspace\left ( x \right )
	   {\mathds{1}_{\widetilde{\Phi}}\left( p_b r^{-\alpha} x \right)} r
	   \mathrm{d}r\,\mathrm{d}x \Bigg\}.
	   \label{EQ:CHARACTERISTIC_FUNCTION_DLPC}
	\end{align}
	And from \eqref{EQ:CUMULANT_PARTIAL_DERIVATIVE} the $n^{\textnormal{th}}$
	cumulant is,
	
	\begin{align}
		\kappa_{n} &= 2\pi\lambda\negthickspace \int\limits_{X} {
		\int\limits_{R_m}^{\min\left [ R_M, \left ( x/{\varrho}_{th}
		\right)^{1/\alpha} \right ]}\hspace{-2em} {(p^{\prime})^n r^{1 - n
		\alpha} x^n } f_X\negthinspace\left ( x
		\right)}\mathrm{d}r\,\mathrm{d}x \nonumber \\
		&= 2 \pi \lambda \left[ \,\,\int\limits_{{\varrho}_M}^{\infty}
		{ \int\limits_{R_m}^{R_M} {(p^{\prime})^n r^{1 - n \alpha} x^n }
		f_X\negthinspace\left ( x \right)}\mathrm{d}r\,\mathrm{d}x +
		\int\limits_{{\varrho}_m}^{{\varrho}_M} { \int\limits_{R_m}^{\left (
		x/{\varrho}_{th} \right)^{1/\alpha}}\negthickspace {(p^{\prime})^n r^{1
		- n \alpha} x^n } f_X\negthinspace\left ( x
		\right)}\mathrm{d}r\,\mathrm{d}x \right],
		\label{EQ:CUMULANT_DLPC}
	\end{align}
	{\noindent where ${\varrho}_m={\varrho}_{th} R_m^\alpha$ and
	${\varrho}_M={\varrho}_{th} R_M^\alpha$}.
	By integrating \eqref{EQ:CUMULANT_DLPC} with respect to $r$, we obtain
	\begin{align}
		\kappa_n &= \frac{2 \pi \lambda\,(p^{\prime})^n}{n\alpha  - 2}
		\bigg\{ \left(R_m^{2-\alpha n}-R_M^{2-\alpha  n}\right)
		\int\limits_{{\varrho}_M}^{\infty} { { x^n } f_X\negthinspace\left ( x
		\right)}\mathrm{d}x + \int\limits_{{\varrho}_m}^{{\varrho}_M} { { \left [
		x^n R_m^{2-n \alpha} - x^{\frac{2}{\alpha}}
		\varrho_{th}^{n-\frac{2}{\alpha}} \right ] } f_X\negthinspace\left ( x
		\right)}\mathrm{d}x \bigg\}.
		\label{a1}
	\end{align}

	Finally, we compute the partial moments of the approximating \ac{LN}
	\ac{RV} $X$ by repeatedly applying Definition \ref{DEF:PARTIAL_MOMENT}, and
	by using the change of variable $X=e^{\mu + \sigma Z}$, where $Z\thicksim
	\mathsf{Normal}\left( 0, 1 \right)$, along with the substitutions
	$\tilde{\varrho}_M = \frac{\ln{\varrho_M - \mu}}{\sigma}$ and
	$\tilde{\varrho}_m = \frac{\ln{\varrho_m - \mu}}{\sigma}$.
	\begin{align}
		\operatorname{E}_X^n{\left [ {\varrho}_M, \infty \right ]} = e^{ n \mu
		+ \frac{n^2 \sigma ^2}{2}} \operatorname{Q}\left[\tilde{\varrho}_M - n
		\sigma\right],
		\label{1}
	\end{align}
	\begin{align}
		\operatorname{E}_X^ \frac{2}{\alpha}{\left [ {\varrho}_m, {\varrho}_M
		\right ]}&= e^{ \frac{2 \mu}{\alpha }+\frac{2 \sigma ^2}{\alpha ^2}}
		\left(\operatorname{Q}\left[\tilde{\varrho}_m-\frac{2 \sigma }{ \alpha
		}\right] - \operatorname{Q}\left[\tilde{\varrho}_M-\frac{2 \sigma
		}{\alpha }\right]\right),
		\label{2}
	\end{align}
	\begin{align}
		\operatorname{E}_X^n{\left [ {\varrho}_m, {\varrho}_M \right ]}&=e^{n
		\mu + \frac{n^2 \sigma ^2}{2}} \left(\operatorname{Q}\left[
		\tilde{\varrho}_m-n \sigma \right] - \operatorname{Q}\left[
		\tilde{\varrho}_M-n \sigma \right]\right),
		\label{3}
	\end{align}
	{\noindent where $\operatorname{Q}[u] = \frac{1}{\sqrt{2\pi}} \int_u^\infty
	e^{-\frac{v^2}{2}} \, \mathrm{d}v$.}
	And by replacing the above expressions in \eqref{EQ:CUMULANT_DLPC},
	\eqref{EQ:CUMULANT_R1} results.
	\section{Proof of Proposition \ref{PROPOSITION:CUMULANT_R2}}
	\label{PROOF:CUMULANT_R2}
	For computing the $n^\mathrm{th}$ cumulant of the aggregate interference for
	interfering picocells in $\mathcal{R}_2$, we, once again, begin
	formulating the corresponding \ac{CF} as
	\begin{align}
	   \Psi_{Z_2} \left ( \omega \right ) = \exp\Bigg\{ 2 \pi \int\limits_{0}^\infty
	   \int\limits_{R_m}^{R_M} \Big[ {}\exp\left ( j w p r^{-\alpha} x \right)
	   - 1\Big]  \lambda f_X\negthinspace\left ( x \right)
	   {\mathds{1}_{\widetilde{\Phi}}^c\left( p_b r^{-\alpha} x \right)} r
	   \mathrm{d}r\,\mathrm{d}x \Bigg\}.
	   \label{EQ:CHARACTERISTIC_FUNCTION_1ST_MARK}
	\end{align}
	{\noindent where $\mathds{1}_{\widetilde{\Phi}}^c\left( p_b r^{-\alpha} x
	\right)$ corresponds to the event of not detecting interfering picocells.}

	The $n^\mathrm{th}$ cumulant is then given by
	\begin{align}
		\kappa_{n} &= 2 \pi \lambda \int\limits_{0}^\infty {
		\int\limits_{\max\left [ R_m, \left ( x/{\varrho}_{th}
		\right)^{1/\alpha} \right ]}^{R_M}\hspace{-2em} {p^n r^{1 - n \alpha}
		x^n } f_X\negthinspace\left ( x \right)}\mathrm{d}r\,\mathrm{d}x
		\nonumber \\
		&=2 \pi \lambda \left[ \int\limits_{0}^{{\varrho}_m} {
		\int\limits_{R_m}^{R_M} {p^n r^{1 - n \alpha} x^n }
		f_X\negthinspace\left ( x \right)}\mathrm{d}r\,\mathrm{d}x +
		\int\limits_{{\varrho}_m}^{{\varrho}_M} { \int\limits_{\left (
		x/{\varrho}_{th} \right )^{1/\alpha}}^{R_M} {p^n r^{1 - n \alpha} x^n }
		f_X\negthinspace\left ( x \right)}\mathrm{d}r\,\mathrm{d}x \right].
		\label{EQ:CUMULANT_1ST_MARK}
	\end{align}
	Similar to the derivation of \eqref{a1}, we first integrate with respect to
	$r$ and obtain
	\begin{align}
		\kappa_n = \frac{2 \pi \lambda\,p^n}{n\alpha  - 2} \bigg\{
		\left(R_m^{2-\alpha n}-R_M^{2-\alpha  n}\right)
		\int\limits_{-\infty}^{{\varrho}_m} { { x^n } f_X\negthinspace\left ( x
		\right)}\mathrm{d}x + \int\limits_{{\varrho}_m}^{{\varrho}_M} { { \left
		[ x^{\frac{2}{\alpha}} \varrho_{th}^{n-\frac{2}{\alpha}} - x^n R_M^{2-n
		\alpha} \right ] } f_X\negthinspace\left ( x \right)}\mathrm{d}x
		\bigg\}.
		\label{a2}
	\end{align}
	And after computing the following partial moment, we obtain the expression
	\eqref{EQ:CUMULANT_R2}.
	\begin{align}
		&\operatorname{E}_X^n{\left [ -\infty, {\varrho}_m \right ]}= e^{ n \mu
		+\frac{n^2 \sigma ^2}{2} }\left( 1 -
		\operatorname{Q}{\left[\tilde{\varrho}_m - n \sigma\right]} \right).
		\label{4}
	\end{align}
\end{appendices}

\bibliographystyle{IEEEtran}
\footnotesize
\bibliography{IEEEabrv,bib/son}

\end{document}